\documentclass{llncs}

\usepackage{amsmath, mathrsfs, amssymb,dsfont}
\usepackage{multirow, booktabs, hhline, enumerate, array}
\usepackage{graphics, geometry,epsfig,epstopdf,subfigure}
\usepackage{pdfsync}
\usepackage{color,cite}
\usepackage{hyperref}

\usepackage{footnote}
\makeatletter
\newcommand\footnoteref[1]{\protected@xdef\@thefnmark{\ref{#1}}\@footnotemark}
\makeatother

\newcommand{\tab}{\hspace{4em}}

\newcommand{\bhdr}[1]{\medskip \noindent\textbf{#1.}}
\newcommand{\ihdr}[1]{\medskip \noindent\textit{#1.}}
\newcommand{\Stp}{\mathscr{S}}
\newcommand{\len}{\lambda}

\newcommand{\remove}[1]{}

\newcommand{\Nds}{\mathds{N}}
\newcommand{\Fds}{\mathds{F}}
\newcommand{\EE}{\mathcal{E}}
\newcommand{\LL}{\mathcal{L}}
\newcommand{\CC}{\mathcal{C}}
\newcommand{\MM}{\mathcal{M}}
\newcommand{\FF}{\mathcal{F}}
\newcommand{\XX}{\mathcal{X}}

\newcommand{\YY}{\mathcal{Y}}
\newcommand{\ZZ}{\mathcal{Z}}

\newcommand{\bset}{\{0,1\}}
\newcommand{\tmin}{{t_{ab}}}
\newcommand{\share}{\mathbf{Share}}
\newcommand{\rec}{\mathbf{Rec}}

\title{Multipath Private Communication: \\ An Information Theoretic Approach}
\author{Hadi Ahmadi, Reihaneh Safavi-Naini}
\institute{Department of Computer Science, University of Calgary, Canada\\
\{hahmadi, rei\}@ucalgary.ca}

\begin{document}

\maketitle

\begin{abstract}
Sending private messages over communication environments under surveillance is a very important challenge in communication security and has attracted a lot of attention from cryptographers through time. We believe that resources other than cryptographic keys can be used to provide communication privacy. We consider private message transmission (PMT) in an abstract multipath communication setting between two communicants, Alice and Bob, in the presence of a third-party eavesdropper, Eve. Alice and Bob have no a priori shared keys and furthermore, Eve is computationally unbounded. There are a total of $n$ paths, and the three parties can have simultaneous access to at most $t_a$, $t_b$, and $t_e$ paths. The parties can reselect their accessed paths after every $\len$ bits of communication over a path. We study two types of perfect (P)-PMT and asymptotically-perfect (AP)-PMT protocols. The former has zero tolerance of transmission error and leakage, whereas the latter allows for positive  error and leakage, which tend to zero as the message length increases. We derive the necessary and sufficient conditions (based on the above parameters) under which  P-PMT and AP-PMT are possible. We also introduce explicit P-PMT and AP-PMT protocol constructions. Our results show that AP-PMT protocols attain much higher information rates than P-PMT ones. Interestingly, Alice and Bob can achieve AP-PMT even in unfortunate conditions that they have the least connectivity ($t_a=t_b=1$) and Eve may access all but one paths ($t_e=n-1$). It remains however an open question whether the derived rates can be improved by more sophisticated AP-PMT protocols.

We study applications of our results to private communication over the real-life scenarios of multiple-frequency links and multiple-route networks. We show  practical examples of such scenarios that can be abstracted by the multipath setting: Our results prove the possibility of keyless information-theoretic private message transmission at rates $17$\% and $20$\% for the two example scenarios, respectively. We discuss open problems and future work at the end.
\end{abstract}
\keywordname Secure message transmission, information-theoretic security, multipath routing, frequency hopping

\section{Introduction}
With the rapid growth of online networks and the Internet, an increasing number of daily activities are moved to the online world, and more communications of individuals fall under prying eyes resulting in increasing loss of privacy. There is a wide range of incentives for capturing personal data and communication by various organizations (including secret agencies). Hackers easily tap into routers or WiFi connections to steal online communication data \cite{Li08}. There are reported news on security agencies watching civilian communications through switches and routers in the Internet \cite{LiRi05,Mc03}. Given massive computational resources that can be accessible through cloud services, na\"{i}ve usage of traditional cryptographic systems for protecting communication  in many cases creates a false sense of security rather than real protection \cite{Le13}. This also has the drawback of limited term security guarantee. Development of quantum algorithms such as Shor's  algorithm \cite{Sh97} will render all today's widely  used crypto algorithms, including RSA-based encryption and Diffie-Hellman key exchange, completely insecure. The widely known one-time-pad with information-theoretic security, on the other hand, requires prior sharing of long keys and is so impractical.

The problem we consider is summarized as follows. There is a communication system, such as a network or a wireless link, that can be massively eavesdropped. The adversary has abundant computation power and is able to break computational crypto algorithms. The communicants are not provided with any pre-shared keys. The question is whether there is any hope for private communication in this setting.

In this paper, we investigate using multiple paths of communication as a resource for providing protection against a computationally-unbounded eavesdropper. A path may have different realizations in different scenarios. It can be a connecting path over an ad-hoc network or Internet, a frequency channel in wireless communication, or a fiber strand in a fiber-optic link. Using  path redundancy in networks for providing security has been considered before in the context of secure message transmission (SMT) \cite{DDWY93}: A sender is connected to a receiver through a set of wires, a subset of which is {\em controlled} by an unlimited adversary. Although the SMT setting allows for different types of adversaries, the focus of research has been security against Byzantine adversaries that completely control the corrupted wires. Because of this strong adversarial model, secure communication is only possible when the number of uncorrupted wires is strictly greater than the number of  corrupted ones \cite{DDWY93}.  This renders SMT impossible in many cases of interest: One cannot assume that in widely eavesdropped networks, communicants are connected by more honest paths than the corrupted ones. This strong adversary framework however is not necessary in surveilled networks where the goal of the adversary is listen in without being noticed. Example of massive network surveillance by security agencies made headline news in most of 2013.

\subsection{Our work: PMT in the multipath setting}
We consider message transmission over the following abstract communication system with three parties: a message sender \textit{Alice}, a message receiver \textit{Bob}, and an eavesdropper \textit{Eve}. Alice wants to send a message to Bob privately, without leaking any information to Eve who is computationally unlimited. There is no shared key between Alice and Bob. The system provides a total of $n$ disjoint paths between Alice and Bob, but not all paths can be accessed simultaneously: Alice and Bob can access to up to $t_a$ and $t_b$ paths at a time, respectively. Eve can also observe communication over up to $t_e$ paths at each moment. We assume time is divided into intervals of equal length $\len$. The parties (including Eve) select their paths at the beginning of each time interval and stay with their choice through the end of the interval. They can freely switch to new paths at the beginning of the next interval. The time interval length is determined by the technological limitations of the parties, in particular the adversary, in switching between paths as well as the communication scenario, e.g., switching between frequencies in a multiple-frequency wireless link can be much faster than corrupting new routers in a computer network.

\textit{We refer to the above communication system as the $(n,t_a,t_b,t_e,\len)$-multipath setting and to the problem as private message transmission (PMT)}. We provide formal definitions for PMT protocols over a multipath setting. Foremost, we are interested in necessary and sufficient connectivity conditions (based on the setting parameters), under which PMT is possible. But we do not stop here. We also study how to attain the highest possible information rate (message bits divided by communicated bits). We call this value the \textit{secrecy capacity}. The study of secrecy capacity and capacity-achieving constructions is quite significant, due to bandwidth limitations and the communication cost in most practical scenarios.

\bhdr{P-PMT and AP-PMT}
We define PMT protocols and measure their security using a reliability parameter ($\delta$) and a secrecy $(\epsilon)$ parameter. The former shows the probability of ``incorrect'' message transmission and the latter represents the message leakage to the adversary. Ideally, a secure PMT protocol is expected to provide perfect reliability $\delta=0$ and perfect secrecy $\epsilon=0$. Relaxing the security requirements to a desired extent may however let PMT protocols send messages at higher rates. We consider designing of two types of PMT protocol families (protocols are indexed based on message length), namely \textit{perfect (P)-PMT} families that include only perfectly-secure protocols and \textit{asymptotically-perfect (AP)-PMT} families that allows for PMTs with positive $\delta$ and $\epsilon$, yet these values tend to zero for protocols that send longer messages. The latter family is particularly interesting because it provides security for a much wider range of connectivity level by admitting an arbitrarily small security loss in protocols. From the security point of view, AP-PMT families behave like P-PMT as the message length tends to infinity.

We define \emph{P-secrecy capacity $C_0$} as well as  \emph{AP-secrecy capacity $C_{\sim 0}$} as the highest achievable rates by P-PMT and AP-PMT families, respectively. We start our investigation of capacity and optimal constructions first for the special case of full-access communication (when $t_a=t_b=n$), and then extend the study to the general case.

\bhdr{PMT in the full-access case}
The problem when $t_a=t_b=n$ is closely related to the traditional study of SMT \cite{DDWY93} when the adversary is only passive. Currently known results can provide partial answers to our questions: Assuming $\len = \infty$, existing SMT results show that P-PMT is obtained by applying a ramp secret sharing scheme (SSS) to the message and sending shares on distinct paths. The necessary and sufficient condition for this is $n > t_e$. However, this does not provide a full answer to the problem. Firstly, as we will show, the solution does not work for finite $\len < \log(2n-t_e)$. Secondly, it remains unclear whether we obtain rate improvement by allowing AP-PMT protocols. We thus complete the study of full-access PMT. Replacing the ramp SSS by an algebraic-geometric quasi-ramp SSS, we show a PMT construction that works when $\len <\log(2n-t_e)$. Rate analysis of this construction together with a capacity upper-bound results proves that the P-secrecy capacity in the full-access case satisfies $1-\frac{t_e}{n} - \Delta \leq C^{\mathbf{FA}}_0 \leq 1-\frac{t_e}{n}$, where $\Delta = (2^{\frac{\len}{2}-2} - 0.25)^{-1}$, for both one-way and two-way communication settings. Although finding the precise capacity expression remains an interesting theoretical question, the bounds give a fairly tight approximation of the capacity for practical values of $\len > 100$ (See Section \ref{sec-practical} on our practical consideration). We next consider AP-PMT and prove the same bounds for AP-secrecy capacity $1-\frac{t_e}{n} - \Delta \leq C^{\mathbf{FA}}_{\sim 0} \leq 1-\frac{t_e}{n}$. The conclusion is
\emph{relaxing secrecy and reliability requirements to asymptotically perfect does not help us improve PMT rates in the full-access case}.

\bhdr{PMT in the general case}
Starting from perfect secrecy, we show that P-PMT is possible if and only if $t_e < \tmin$, where $\tmin=\min(t_a,t_b)$. Furthermore, the capacity falls between the almost tight bounds of $1-\frac{t_e} \tmin  - \Delta \leq C_0 \leq 1-\frac{t_e}{\tmin}$, regardless of whether one-way or two-way communication is allowed. Note that the bounds are irrespective of $n$: The PMT construction simply fixes $\tmin$ paths (and forgets about the rest of the network) and applies SSS as in the full-access case. When $\tmin \ll n$, the capacity $C_0 \approx 1-\frac{t_e}{\tmin}$ is far below the full-access capacity $C^{\mathbf{FA}}_0 \approx 1-\frac{t_e}{n}$. We show however that higher rates can be achieved by using AP-PMT instead. We show that $t_e < t_b$ is the necessary and sufficient connectivity condition for ``one-way'' AP-PMT; more surprisingly, we introduce a construction which achieves the rate $\frac{1-\frac{t_e}{n}-\Delta}{1 + \xi^*}$ for some small constant $\xi^*$. Numerical results show that for many interesting scenarios, the rate is fairly close to the full-access capacity $C^{\mathbf{FA}}_0$. Last but not least, we show that when interactive communication is allowed, AP-PMT is possible even when $t_e \geq t_b$, only requiring (i) Eve does not observe all paths ($t_e < n$) and (ii) Alice and Bob can communicate ($t_a,t_b>0$). We introduce an interactive PMT construction with the secrecy rate $\frac{1-\frac{t_e}{n}-\Delta}{1 + \xi^{**}}$, for some small constant $\xi^{**}$ that decreases by $\len$. When $\len$ is large enough, AP-PMT rates get close to $C^{\mathbf{FA}}_0$ (as in full-access) even in situations where communicants suffer from poor connectivity regime, i.e., $t_e=n-1$ and $t_a=t_b=1$.

Tables \ref{tab-connectivity-capacity} summarizes our theoretical results about P-PMT and AP-PMT over the multipath setting. For simplicity, we assume that $\Delta \to 0$ which is reasonable according to our practical consideration of $\len > 100$. While connectivity range remains always the same in the full-access case, in general, AP-PMT works under a much wider connectivity range compared to P-PMT. When interaction is permitted, AP-PMT is always possible as if in the full-access case.

\begin{table}
\centering
\caption{PMT connectivity conditions and capacities in the $(n,t_a,t_b,t_e,\len)$-multipath setting}\label{tab-connectivity-capacity}
\begin{tabular}{c|c|cc|c|c|}
\multicolumn{2}{c}{}        & \multicolumn{2}{c}{\bf Full Access} & \multicolumn{2}{c}{\bf Partial Access} \\
\cline{3-6}
\multicolumn{2}{c|}{}         & One-way & \multicolumn{1}{|c|}{Two-way} & One-way & Two-way \\
 \cline{2-6}
\multirow{2}{*}{\bf Connectivity~} & \bf P-PMT  & \multicolumn{2}{c}{\multirow{2}{*}{$t_e < n$}} & \multicolumn{2}{|c|}{$t_e < \tmin$}  \\
 \cline{2-2} \cline{5-6}
& \bf AP-PMT  &                         &                       ~&~ $t_e < t_b$ \& $t_a > 0$ ~&~ $t_e < n$ \& $t_a,t_b > 0$\\
\hhline{~=====}
\multirow{2}{*}{\bf Capacity} & $\mathbf{C_0}$  & \multicolumn{2}{c}{\multirow{2}{*}{$1-\frac{t_e}{n}$}} & \multicolumn{2}{|c|}{$[1-\frac{t_e}{\tmin}]_+$}  \\
 \cline{2-2} \cline{5-6}
& $\mathbf{C_{\sim 0}}$  &        &         ~&~ $\frac{1-\frac{t_e}{n}}{1+\xi^{*}} < C_{\sim 0} < 1-\frac{t_e}{n}$ ~&~ $\frac{1-\frac{t_e}{n}}{1+\xi^{**}} <C_{\sim 0} < 1-\frac{t_e}{n}$\\
 \cline{2-6}
\end{tabular}
\end{table}

Comparing the P-secrecy and AP-secrecy capacities shows that in the full-access case, both equal $1-\frac{t_e}{n}$ and the rate is achieved by the one-round PMT scheme $\FF_0$. In general however, the achievable rates for P-PMT and AP-PMT deviate: While P-PMT rates cannot exceed $[1-\frac{t_e}{\tmin}]_+$, it is possible to get close to the maximum rate of $1-\frac{t_e}{n}$ (as in full-access) by taking benefit of AP-PMT protocols. The values $\xi^{*}$ and $\xi^{**}$ vary depending on the setting parameter. The precise values are given according to Corollaries \ref{corollary-oneway-AP-cap} and \ref{corollary-twoway-AP-cap}, respectively. In Section \ref{subsec-rateCompare}, we make clear how close these lower-bounds are to the upper-bound.

It remains a theoretically interesting open question to derive the exact expression for the AP-secrecy capacity in the general multipath setting.

\bhdr{Practical consideration and numerical analysis}
To show the relevance of our results to practical communication, we consider two example scenarios of private communication over (i) multiple-frequency links and (ii) multiple-route networks. In both scenarios, we elaborate on practically sound examples that can be modeled by the multipath setting and for which private communication is promised at rates $17$\% and $20$\%, respectively. To our knowledge, this is the first attempt to build optimal rate communication with information-theoretic privacy in these scenarios.

\subsection{Related work}

\subsubsection{Secure message transmission}
The secure message transmission (SMT) problem has been first studied by Dolev et al. \cite{DDWY93}: Alice and Bob are connected by $n$ paths, out of which $t \leq n$ can be corrupted  by the active adversary, Eve. The objective is to guarantee both privacy and reliability of a transmitted message from Alice to Bob. SMT was initially motivated for being use in secure communication and multi-party computation (MPC) protocols \cite{BGW88}, and more recently has found applications in key generation over distributed networks \cite{WS11}. Despite similarities, our work deviates from the SMT literature in a few directions. Firstly, we focus on passive adversaries. This allows us to investigate in more details capacity-achieving constructions for P-PMT and AP-PMT. More importantly as we mentioned earlier, a great portion of threats to online communication are based on passive attacks, trying to violate personal data privacy. Although SMT results against active adversaries can also be utilized for PMT, this is to be consider a security over-design, which most likely leads to inefficient and sub-optimal solutions.

Secondly, SMT assumes Alice and Bob can access to all $n$ paths, i.e., only Eve is limited in her connectivity. This assumption is not always practical. In dense networks with thousands of nodes and paths, there may be more communication paths than parties can possibly afford access simultaneously. Another example, an optical fiber for instance can carry information over up to thousands of wavelengths, but it becomes too expensive to communicate data over all frequencies simultaneously.

The last difference here is the concept of time interval. The SMT work assumes that Alice and Bob can communicate arbitrarily many bits in each round without Eve being able to switch her corrupted paths. This is not a practical assumption. In a real-life scenario, Eve may move to new paths if enough time is provided. We capture this by adding a time-interval length parameter $\len$ to our abstract model. As a consequence, an SMT protocol that transmits more than $\len$ bits in a round without accounting for Eve's movements is not necessarily functional in the new model.

\subsubsection{Frequency hopping}
Spread spectrum is a data transmission technique which spreads the signal energy over a relatively wide frequency range rather than a certain band. The advantages of spread spectrum include low probability of intercept, high resistance to noise and jamming, and little interference with conventional signal transmissions over the same range, due to low power density at each frequency.  Frequency-hopping spread spectrum (FHSS) is  a special form of spread spectrum suited for low-rate and low-power systems. The technique transmits data as a sequence of blocks, where each block is sent over a frequency channel that is selected (pseudo-randomly) from a pool of frequency channels. The FHSS technology has appeared in walkie-talkie devices, early WiFi, homeRF, and Bluetooth. These applications use frequency hopping to enhance communication quality by reducing the effect of interference and narrow-band noise. More recently, the technique has been used to countermeasure jamming-based denial of service (DoS) attacks over hostile communication environments such as sensor networks \cite{WoSt02}. FHSS typically requires the sender and the receiver to be coordinated, i.e., to have pre-shared keys which let them choose a sequence of random frequencies to communicate over. This requirement, however, cannot be always fulfilled specially in ad hoc networks, where two nodes may happen to communicate for the first time. Strasser et al. \cite{SCPC08} addressed this concern by introducing uncoordinated frequency hopping (UFH) and studying its application to jamming-resistant key establishment. The work encouraged several follow-up works, where UFH was examined for objectives such as jamming-resistant broadcast communication \cite{PSC09,LNDL10}.

Although the UFH technique provides key-less secure communication over physical channels, the security it offers is solely ``jamming resistance''. This should not be confused with data confidentially and/or integrity that is traditionally defined in cryptography. UFH does not guarantee data confidentiality and assumes it has been taken care of over higher layers via cryptographic tools. UFH does not provide data integrity either: It even takes use of higher-layer message authentication to provide reliable assembly of transmitted packets as a countermeasure against jamming attacks. Relying on higher-layer cryptographic primitives in UFH has two drawbacks: (i) the need for a public-key infrastructure (or shared keys) and (ii) only computational security guarantees.

We for the first time show that multiple-frequency channels can be used for private communication. In contrast to the above, our approach to PMT has two advantages: (i) it does not rely on higher-layer cryptography, and (ii) it provides security against computationally-unlimited adversaries.

\subsubsection{Multipath routing}
Multipath routing consists of finding multiple routes over a network from a source node to a destination node, and using them for the purposes of  reliable data transmission, load balancing, and/or higher aggregate bandwidth. The technique has has been explored in different contexts, such as the Internet, sensor networks, and mobile ad-hoc networks (MANETs). In all these applications, multipath routing includes two main components, namely \textit{route discovery and maintenance} and \textit{traffic allocation}. The former finds the multiple routes between the communicants, and the latter deals with how data is distributed amongst the routes. Split multipath routing (SMR) \cite{LeGe01} is one of such protocols, proposed for MANETs, that allows the source node to discover and allocate traffic over maximally-disjoint paths.

Despite multipath routing approaches have been mainly proposed for reliability and efficiency of data transmission, the idea can be also used to provide security. We noted earlier that a primary motivation for the well-studied SMT problem is secure communication over highly connected networks. Similarly, one can motivate the application of our results to private transmission over such networks. The advantage of PMT is that it allows Alice and Bob to communicate privately even if they do not access to ALL disjoint paths, rather use a random subset that is not known to the eavesdropper.

\subsection*{Notation}
We use $[x]_+$ for a real value $x$ to show $\max\{0,x\}$. For two random variables $X \in \XX$ and $Y \in \YY$, we denote their statical distance by $SD(X,Y)=0.5 \sum_{x \in \XX} |\Pr(X=x)-\Pr(Y=x)|$. All logarithms are in base 2. The following notations are specific to our work and are used frequently throughout the paper.

\medskip
\begin{tabular}{!{\vrule width 2pt} c|l !{\vrule width 2pt} c | l !{\vrule width 2pt}}
 \noalign{\hrule height 2pt}
$n$ 		& Total number of communication paths    ~& $\Stp$ & A multipath setting \\
$t_a$ 	& Number of paths accessible to Alice        & $\Pi$ & A PMT protocol\\
$t_b$ 	& Number of paths accessible to Bob          &~ $\FF$~ & A PMT scheme or a family of protocols~ \\
$t_e$  	& Number of paths accessible to Eve          & $R$ & Secrecy rate    \\
$\len$ 		& Length of a time interval in bits      & $C$ & Secrecy capacity \\
~$\tmin$~ 		& $\min(t_a,t_b)$                    & $\Delta$ & $(2^{\frac{\len}{2}-2} - 0.25)^{-1}$ \\
 \noalign{\hrule height 2pt}
\end{tabular}
\medskip

\noindent Throughout, we consider $\Delta$ a negligible value for our numerical analysis by assuming large $\len$.

\section{Preliminaries}
\subsection{Threshold, ramp, and quasi-ramp secret sharing schemes}\label{subsec-SSS}
Secret sharing is the task distributing a secret value $S$ among a set of $m$ players such that only qualified subsets of those players can recover the secret efficiently, while no information is leaked to an unqualified subset. A secret sharing scheme (SSS) is defined by a pair $(\share,\rec)$ of functions. The share function $\share$ maps secret $S$ to shares $X_1, X_2,\dots,X_m$, and the recovery function $\rec$ maps the presented shares $X'_1, \dots, X'_m$ to a secret estimate $\hat{S}$. At the time of reconstruction, the $i$-th share value $X_i$ can be null (showed by $\Lambda$) meaning that the player $i$ is not present. The secret estimate $\hat{S}$ is expected to equal $S$ if the collection of shares is qualified and $\bot$ otherwise.

A $(k,m)$-threshold secret sharing scheme (SSS) \cite{Sh79,Bl79} distributes a secret as $m$ shares such that any $\geq k$ shares are qualified and any $k-1$ or fewer are unqualified. Ramp SSSs are thus introduced as an extension to relax the secrecy requirement aiming at lowering the share size. As an extension of this, a $(k,r,m)$-ramp SSS where $r\leq k$, guarantees that $\geq k$ shares are qualified and $\leq k-r$ shares are unqualified, while information leakage increases as the number of shares tends from $k-r$ to $k$. The definition gives threshold SSS as a special case when $r=1$.

\bhdr{Polynomial-based SSS}  The simplest example of a $(k,n)$-threshold SSS is the polynomial-based construction, due to Shamir \cite{Sh79}, which puts the secret (from field $\Fds_p$ of size $p$) as the constant term of a random polynomial of degree $k-1$ over $\Fds_p[x]$ and obtains $m$ shares as points on the polynomial. The construction can be easily converted to a ramp SSS by allowing the secret to include $r$ (instead of $1$) points on the random polynomial. Both threshold and ramp schemes are optimal in their kind as they provide the smallest possible share sizes achievable for the secrecy that they promise. The detailed description of the polynomial-based $(k,r,m)$-ramp SSS is as follows.

\ihdr{$(k,r,m)$-ramp SSS $(\share_{pol},\rec_{pol})$} Let $p \geq m+r$ and $S \in \Fds^r_p$ be the secret.
\begin{itemize}
\item $\share_{pol}(S)$ chooses a random polynomial $f(x)$ of degree $k-1$ over $\Fds_p[x]$, such that $f(0)=S_0, f(1)=S_1, \dots, f(r-1) = S_{r-1}$ and returns $m$ shares $X_1=f(r), X_2=f(r+1), \dots, X_m = f(r+m-1)$.
\item $\rec_{pol}(X'_1,\dots,X'_m)$ chooses the first $k$ present shares $X_i \neq \Lambda$ (if not possible returns $\bot$), obtains $f(x)$ through Lagrange interpolation, and returns the secret $S=(f(0), f(1), \dots, f(r-1))$.
\end{itemize}

\bhdr{Algebraic-geometric SSS} In the polynomial-based SSS, the number of shares ($m$) should be less than the number of distinct points $p$ on the polynomial, i.e., no more than $p-1$ shares can be distributed. Algebraic-geometric constructions resolve this issue by using curves of high enough genus, instead of polynomials, over fields. Garcia and Stichtenoth \cite[Theorem 3.1]{GaSt96} show an explicit family of curves over a field of size $p$ (when $p$ is a square) which have at least $(\sqrt{p}-1) g$ points, where $g$ is the curve genus. Chen and Cramer \cite{ChCr06} use this result to construct an algebraic geometric $(k,g,m)$-quasi-threshold SSS over $\Fds_p$ for any $m < (\sqrt{p}-1)g$ for arbitrarily large $g$: quasi-threshold SS guarantees any $k-1$ or fewer shares are unqualified and any $k+2g$ or more shares are qualified. As the authors note \cite[Section 4]{ChCr06} the construction can be used to give a $(k,r,g,m)$-quasi-ramp SSS, where $k+2g\leq m$, with $\leq k-r$ shares unqualified and $\geq k+2g$ shares qualified. The following is a description of Chen and Cramer's algebraic-geometric SSS using Garcia-Stichtenoth family of curves.

\ihdr{$(k,r,g,m)$-quasi-ramp SSS $(\share_{alg},\rec_{alg})$} Let $\CC$ be a Garcia-Stichtenoth curve with genus $g$ over $\Fds_p$, where $p$ is a square and $(\sqrt{p}-1) g \geq m+r$. Define $Q$, $P_0, P_1, \dots, P_{m+r-1}$ as any $m+r+1$ distinct rational points on $\CC$, $D=(k+2g).(Q)$ as a rational divisor of $\CC$, and $\LL(D)$ as the Riemann-Roch space associated with $D$. Let $S \in \Fds^r_p$ be the secret.
\footnote{Refer to \cite{ChCr06} for the definitions of rational divisor and Riemann-Roch space.}
\begin{itemize}
\item $\share_{alg}(S)$ chooses a random function $f(.) \in \LL(D)$ such that $f(P_0)=S_0, f(P_1)=S_1, \dots$, $f(P_{r-1}) = S_{r-1}$ and returns $m$ shares $X_1=f(P_r), \dots, X_m = f(P_{m+r-1})$ over $\Fds_p$.
\item $\rec_{alg}(X'_1,\dots,X'_m)$ chooses the first $k+2g$ non-null shares $X_i \neq \Lambda$ (if not possible returns $\bot$), obtains $f(.)$ through linear interpolation and returns the secret $S=(f(0), f(1), \dots, f(r-1))$.
\end{itemize}

The important advantage of the above construction is the improved condition $m+r \leq (\sqrt{p}-1) g$ such that the right hand size can be large with $g$, despite a constant field size $p$. This however comes with the price of adding an extra $2g$ gap between the number of qualified and unqualified players. This gap is not generally desired as it implies less security (security against fewer players) if one sticks to a fixed number of qualified players. Although both factors above are linear with respect to $g$, for large enough field size $p$, the benefit of increasing $g$ is dominant: one can generate $\sqrt{p}-1$ additional shares by allowing an extra $2$-player gap in SSS. We take use of this interesting property in our PMT constructions in which we will have $p=2^\len$, where $\len$ is the time interval length.

\section{Problem Description}
\subsection{Multipath setting abstraction}
A \emph{partial-access multipath communication setting} (or \emph{multipath setting} in brief) refers to an abstract communication system which consists of $n$ disjoint communication paths, out of which at most $t_a$, $t_b$, and $t_e$ paths can be accessed by Alice, Bob, and Eve, respectively, at any point in time. More precisely, time is divided into equal-length intervals, each of which corresponding to sending $\len$ consecutive bits over at least one path by either Alice or Bob. The parties choose their sets of access paths at the beginning of each time interval and will hold on to their choice till the end of that interval, i.e., until $\len$ bits are communicated over at least one of the paths. This abstraction of time intervals in bits is obtained by multiplying the bit-transmission speed by path switching time. To summarize, a multipath setting is defined by five public parameters of $(n,t_a,t_b,t_e,\len)$. It is implicit that $t_e,t_a,t_b \leq n$. We always use $\tmin=\min(t_a,t_b)$ to denote the minimum number of paths that both Alice and Bob can access. For the special case of $t_a=t_b=n$, we say the setting provides ``full access'' to Alice and Bob and refer to the setting as an $(n,t_e,\len)$-full-access setting. This is in line with the existing SMT work \cite{DDWY93} which assumes Alice and Bob are not limited in their access. Figure \ref{fig-multipath} illustrates full-access versus partial-access settings.

\begin{figure}[hbt]
  \centering
      \subfigure[Full-access: $t_a=t_b=n$.]{
  \includegraphics[width=.43\textwidth]{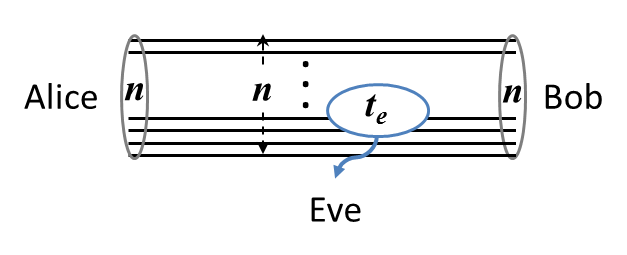}
    \label{fig-TFsetting}
    }
      \subfigure[Partial-access: $t_a,t_b \leq n$.]{
  \includegraphics[width=.43\textwidth]{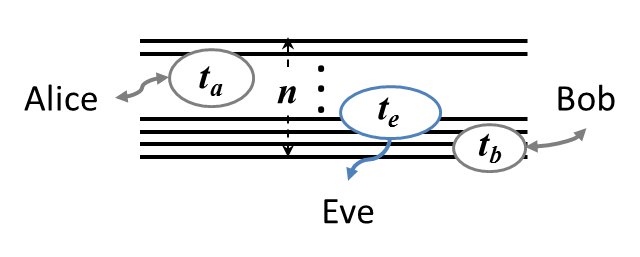}
    \label{fig-TPsetting}
    }
    \caption{Full-access vs. partial-access multipath settings}\label{fig-multipath}
\end{figure}

In general, the value of $\len$ depends on how fast the communicants and (more importantly) Eve can release old paths and capture new paths without possibly missing the live communication. This relates to the actual communication scenario, the communication capability of the transceiver devices, and the transmission speed. We shed more light on this in Section \ref{sec-practical}: In the practical scenarios considered there, the least interval length $\len = 52$ corresponds to the multiple-frequency link scenario, and is obtained by multiplying the $1 \mu$s switching time between frequencies and the transmission speed of $100$ Mbps. Thus for our numerical analysis, we always assume $\len > 100$.

\subsection{PMT protocol and secrecy capacity: definition}
To send a message $S \in \bset^k$ securely from Alice to Bob, a PMT protocol allows them to communicate a total of $c$  bits on their accessed paths (possibly back and forth in multiple rounds) so that Bob computes a variable $\hat{S}$ as his estimate of the message $S$. Eve will obtain the view $View_E(S)$ of the communication and uses it to obtain some knowledge about $S$. The randomness in $View_E(.)$ comes from the randomness of the PMT protocol and that of the adversary.

\begin{definition}[PMT Protocol] \label{def-PMT}
The protocol $\Pi$, as described above, over a multipath setting is a $(k, c, \delta, \epsilon)$-PMT protocol if it transmits any $k$-bit message using $c$ bits of communication such that
\footnote{The definition looks simpler than that in the SMT literature \cite{DDWY93}, since it only considers passive adversary.}
\begin{eqnarray}
& \mbox{Reliability}:&~ \forall s\in \bset^k:~~ \Pr(\hat{S} \neq s) \leq \delta, \label{reli} \\
& \mbox{Secrecy}:&~ \forall s_1,s_2 \in \bset^k:~~ SD\left(View_E(s_1), View_E(s_2) \right) \leq \epsilon. \label{sec}
\end{eqnarray}
The \emph{secrecy} rate of $\Pi$ is obtained as $R=\frac{k}{c}$. The protocol is called perfectly reliable when $\delta=0$, perfectly-secret when $\epsilon=0$; if both hold, $\Pi$ is called a perfectly (P)-PMT protocol.
\end{definition}

In practice, the message length may be unknown before hand and one needs a family of PMT protocols that can be used for arbitrarily long messages. Good PMT families are expected to have a guaranteed rate. We define the optimality of a PMT family by its secrecy rate which is measured as the minimum of the rates of all protocols it includes. We define $(\delta,\epsilon)$-, perfect (P)-, and asymptotically-perfect (AP)- PMT families.

\begin{definition}[$(\delta,\epsilon)$-PMT and P-PMT families]\label{def-de-SMTfamily}
A $(\delta,\epsilon)$-PMT family $\FF$ for a multipath setting $\Stp$ is an infinite sequence $(\Pi_i)_{i\in \Nds}$, where for each $i\in \Nds$, $\Pi_i$ is a $(k_i,c_i,\delta,\epsilon)$-PMT protocol over $\Stp$ and $k_{i+1} > k_i$.
\\
The $(\delta,\epsilon)$-secrecy rate of $\FF$ is defined as \footnote{The infimum exists as the sequence is bounded from below by zero.}
\[R_{\FF:\delta,\epsilon}  = \inf \{ \frac{k_i}{c_i}: ~ i \in \Nds \}.\]
When $\delta=\epsilon=0$, $\FF$ is called a perfect (P)-PMT family and the P-secrecy rate is denoted by $R_{\FF:0}$.
\end{definition}

P-PMT families include only perfectly-secure protocols (with $\delta=\epsilon=0$). Designing such families of protocols is important for highly-sensitive data transmission where absolutely zero failure and leakage is acceptable. There are however scenarios which desire non-zero yet arbitrarily small $\delta$ and $\epsilon$, i.e., one expects to pick a protocol with $\delta$ and $\epsilon$ as close as desired to zero. We hence define AP-PMT families which include  $(\delta,\epsilon)$-PMT protocols; however, by choosing to send longer messages, the values $\delta$ and $\epsilon$ decrease and tend to zero asymptotically.

\begin{definition}[AP-PMT family]\label{def-asymp-SMTfamily}
An AP-PMT family $\FF$ for a multipath setting $\Stp$ is an infinite sequence $(\Pi_i)_{i\in \Nds}$ where for each $i\in \Nds$, $\Pi_i$ is a $(k_i,c_i,\delta_i,\epsilon_i)$-PMT protocol over $\Stp$ and it holds that $k_{i+1} > k_i$, $\delta_{i+1} \leq \delta_i$, $\epsilon_{i+1} \leq \epsilon_i$, and $\lim_{i \to \infty} \delta_i = \lim_{i \to \infty} \epsilon_i =0$.

\noindent The \emph{AP-secrecy rate} of $\FF$ is defined as
\[R_{\FF:\sim 0} = \inf \{ \frac{k_i}{c_i}: ~ i \in \Nds \}.\]
\end{definition}

We shall now define the different types of secrecy capacities for a multipath setting $\Stp$. The capacity means the highest secrecy rate that can be guaranteed for all message lengths by a PMT protocol.
\begin{definition}[Secrecy Capacity]\label{def-secrecy-capacity}
The $(\delta, \epsilon)$- (resp. P- and AP-) secrecy capacity $C_{\delta, \epsilon}$ (resp. $C_0$ and $C_{\sim 0}$) of a multipath setting $\Stp$ equals the largest $(\delta, \epsilon)$- (resp. P- and AP-) secrecy rate achievable via all possible $(\delta, \epsilon)$- (resp. P- and AP-) PMT families over $\Stp$.
\end{definition}

\subsection{Relation among secrecy capacities}
Definition \ref{def-PMT} implies that any $(k,c, \delta_1, \epsilon_1)$-PMT protocol is also $(k,c, \delta_2, \epsilon_2)$-PMT for $\delta_2\geq \delta_1$ and $\epsilon_2\geq \epsilon_1$. Inasmuch as families simply consist of protocols, this suggests that an $(\delta_1, \epsilon_1)$-PMT family is also a $(\delta_2, \epsilon_2)$-PMT family . Put differently, the set of all $(\delta_2, \epsilon_2)$-PMT families is a superset of the set of all $(\delta_1, \epsilon_1)$-PMT families. The conclusion is that $C_{\delta_2, \epsilon_2} \geq C_{\delta_1, \epsilon_1}$, and more generally, the capacity $C_{\delta, \epsilon}$ decreases as $\delta$ and $\epsilon$ decrease and is lower bounded by $C_0$. Using a similar argument about the relation with AP-PMT families, one can reach at
\[C_{0} \leq C_{\sim 0} \leq C_{\delta, \epsilon}, ~~\mbox{for}~ \epsilon>0 ~\mbox{and}~ \delta>0.\]
We ask whether the above inequalities can be replaced by equality. Showing cases for which $C_{\sim 0} < C_{\delta, \epsilon}$ is not ambitious: Consider for instance when $t_a=t_b=t_e = n$; for $\epsilon=1$ since no secrecy is expected, Alice can achieve reliable transmission at full rate $C_{\delta,1}=1$, but requiring arbitrarily small $\epsilon$ leads to zero rate $C_{\sim 0}=0$. The challenging question is, however, the equality of $C_{0}$ and $C_{\sim 0}$. Having an answer to this question is essential: It is fairly reasonable to tolerate arbitrarily small deviation from perfect security in order to improve rate or to convert impossibility of PMT to possibility. In the rest of the paper, we study P-PMT and AP-PMT protocols starting from the special case of full-access communication (when $t_a=t_b=n$), and extending it to the general multipath setting. Our study leads us the following ultimate conclusion:

\begin{center}
{\it For a wide range of (partial-access) multipath settings, it holds that $C_{\sim 0} > C_0$, i.e., AP-PMT results in higher rates than P-PMT.}
\end{center}

\section{PMT in the full-access scenario}
In the full-access multipath setting, Alice and Bob have access to all paths $t_a=t_b=n$; for simplicity, we refer to this as the $(n,t_e,\len)$-full-access setting. Considering this special case when $\len = \infty$ (i.e. Eve cannot switch her paths in a round once they are chosen), the PMT problem relates to the SMT work (for passive adversary), where the optimal solution simply uses a polynomial-based $(n,r,n)$-ramp secret sharing scheme (SSS), denoted by $(\share_{pol},\rec_{pol})$, where $r=n-t_e$. The obtained PMT solution, denoted by $\FF^{pol}_0$, is described below. Let $S=(S_1,\dots,S_r) \in \Fds_{2^u}^r$ be the secret message to be transmitted, for arbitrary $u>\log(2n-t)$.

\bhdr{Polynomial-based P-PMT scheme $\FF^{pol}_0$}
\begin{itemize}
\item Alice calculates shares $(X_1,X_2,\dots,X_n)=\share_{pol}(S)$ and sends share $X_i$ over the $i$-path (path indices are fixed and public).
\item Having received $X_i$'s, Bob obtains the message as $S=\rec_{pol}(X_1,X_2,\dots,X_n)$.
\end{itemize}
Thanks to the SSS secrecy guarantees, Eve obtains no information from her $t$ shares.

\begin{theorem}\label{theorem-FF0}
The scheme $\FF^{pol}_0$ gives a family of $(ur,un, 0,0)$-P-PMT protocols with secrecy rate $R_{\FF^{pol}_0}=1-\frac{t_e}{n}$ over the $(n,t_e,\len)$-full-access setting with $\len = \infty$.
\end{theorem}
\begin{proof}
See Appendix \ref{app-proof-FF0}
\end{proof}

\subsection{P-PMT for finite $\len$}
When $\len$ is finite, the scheme $\FF^{pol}_0$ (without any modification) does not provide us with a P-PMT family since it cannot give protocols for arbitrarily message lengths, rather only for $u \leq \len$: Recall that the message consists of $r$ field elements in $\Fds_{2^u}$, and the SSS shares are from $\Fds_{2^u}$. If $u > \len$, a share $X_i$ needs to be sent in more than one time interval. This lets Eve switch her paths and get (partial) information about more than $t$ shares to learn some information about $S$. However, there is an easy fix to this. To circumvent the above issue, we can always stay with a constant field size $2^\len$ so that each share is delivered in a ``single'' time interval, i.e., it will be either completely leaked or perfectly secure (note that Eve has synchronous path switching with communicants). We instead repeat $\FF^{pol}_0$ for sufficiently many times (time-intervals) to send arbitrarily long messages; hence, a PMT family.

The situation is however more unfortunate for settings where $2n-t > 2^{\len}$, since $\FF^{pol}_0$ cannot provide even a single PMT protocol: The polynomial-based SSS requires $2n-t$ polynomial points ($n-t$ for message $S$ and $n$ for shares $X_1$ to $X_n$); hence the field size $2^u$ should be at least $(2n-t)$. This cannot happen due to the inequalities $2^u < 2^{\len} < 2n-t$. Proposition \ref{proposition-FF0} concludes our result about $\FF^{pol}_0$.

\begin{proposition}\label{proposition-FF0}
Repeating $\FF^{pol}_0$ for arbitrary $q \in \Nds$ times results in a P-PMT family with secrecy rate $R_{\FF^{pol}_0}=1-\frac{t_e}{n}$ over any $(n,t_e,\len)$-full-access setting that satisfies $2n-t \leq 2^{\len}$.
\end{proposition}

\begin{remark}
The condition $2 n - t_e < 2^\len$ holds for all scenarios of interest in practice: We suppose there cannot be more than a ``million'' paths ($2n-t_e < 2^{20}$) and switching between paths cannot be done before $\len=50$ bits of communication. It remains a theoretical question whether P-PMT is possible when the condition is not satisfied.
\end{remark}

When $n-2t_e > 2^\len$, we show that P-PMT families can be made by using an algebraic-geometric quasi-ramp SSS instead of the polynomial-based one: The replacement allows us to have arbitrarily many shares while staying with a constant field size $2^\len$. We define $q \in \Nds$ as a factor to introduce increment to the message length. The scheme requires algebraic-geometric SSS with $q (2n-t_e)$ point evaluations, implying using curve genus $g$ such that $(\sqrt{2^\len} -1) g \geq q (2n-t_e)$ (see Section \ref{subsec-SSS}). Let $g = \lceil \frac{q (2n - t_e)}{2^{\len/2}-1} \rceil$ and $r=q (n - t_e) -2g$. Let $S=(S_1,\dots,S_{q r}) \in \Fds_{2^\len}^{r}$ be the message to be transmitted. We design PMT scheme $\FF^{alg}_0$ similarly to $\FF^{pol}_0$, except it uses an $(q n- 2g,r,g, q n)$-quasi-ramp SSS $(\share_{alg},\rec_{alg})$ over field $\Fds_{2^\len}$ for secret sharing and sends the message in $q$ time intervals.

\bhdr{Algebraic-geometry P-PMT scheme $\FF^{alg}_0$}
\begin{itemize}
\item Alice calculates $q n$ shares $\underline{X}=\left( X_{i,j}\right)_{1\leq i \leq q, 1\leq j \leq n} = \share_{alg}(S)$ and sends each $n$-share subsequence $\left(X_{i,j}\right)_{1\leq j \leq n}$ at time interval $i$ over $n$ randomly chosen paths.
\item Bob receives all $q n$ shares and calculates $S=\rec_{alg}(\underline{X})$.
\end{itemize}

The secrecy and rate analysis of $\FF^{alg}_0$ can be done similarly to that of $\FF^{pol}_0$. The reconstruction and secrecy properties of the quasi-ramp SSS imply that Bob can retrieve the message from his received $q n-2g+2g = q n$ shares and Eve obtains no information based on her view of $q n-2g-r = q t_e$ shares. The secrecy rate here is slightly lower than $\FF^{pol}_0$:
\begin{eqnarray*}
R_{\FF^{alg}_0} = \frac{r \len}{q n \len} = \frac{q(n-t_e)-2g}{q n} = 1-\frac{t_e}{n} - \Delta,
\end{eqnarray*}
where (inequality (a) is obtained by choosing $q \geq \frac{2^{\len/2}-1}{t_e}$)
\begin{eqnarray}\label{Delta}
\Delta = \frac{2\lceil \frac{2 q n}{2^{\len/2}-1} - \frac{q t_e}{2^{\len/2}-1} \rceil}{q n} \stackrel{(a)}{\leq} \frac{2\lceil \frac{2 q n}{2^{\len/2}-1} -1 \rceil}{q n} \leq \frac{4}{2^{\len/2}-1} = (2^{\frac{\len}{2}-2}-0.25)^{-1}.
\end{eqnarray}

\begin{corollary}\label{cor-FF-0-full-access}
The scheme $\FF^{alg}_0$ gives a P-PMT family over any $(n,t_e,\len)$ full-access setting. The P-secrecy rate of the family is $R_{\FF^{alg}_0:0} = [1-\frac{t_e}{n} - \Delta]_+$ where $\Delta \leq (2^{\frac{\len}{2}-2}-0.25)^{-1}$.
\end{corollary}

\subsection{Implication to P-secrecy capacity}
The existing work on SMT (cf. \cite{PCPS10}) suggests the upper-bound $1-\frac{t_e}{n}$ on achievable P-PMT secrecy rates. This combined with the above results leads us to the following approximation of the P-secrecy capacity for the full-access case:
\begin{eqnarray}\label{bounds-C-FA-0}
[1- \frac{t_e}{n} -(2^{\frac{\len}{2}-2}-0.25)^{-1}]_+ \leq  C^{\mathbf{FA}}_0 \leq 1- \frac{t_e}{n}.
\end{eqnarray}

It remains an interesting theoretical question to close the gap between the two bounds. For practical scenarios ($\len > 100$) however, the gap $(2^{\frac{\len}{2}-2}-0.25)^{-1}$ is reasonably close to zero.

\subsection{AP-PMT in the full-access case}\label{subsec-AP-full}
We have so far derived bounds on the P-secrecy capacity in the $(n,t_e,\len)$-full-access setting. We are interested in finding whether rates can be improved if reliability or secrecy requirements are relaxed to asymptotically perfect. To find an answer, we shall obtain the relation between $(\delta,\epsilon)$-PMT protocols and P-PMT ones and the study this relation when $\delta$ and $\epsilon$ approach to zero.

\bhdr{Upper-bounding $(\delta,\epsilon)$-secrecy rates via secret-key rates} There may be different ways to upper-bound the secrecy rates achievable by $(\delta,\epsilon)$-PMT protocols. We here provide a unique approach by upper-bounding secret-key rates of $(\delta,\epsilon)$ secret-key establishment (SKE) protocols. A SKE protocol is defined similarly to a PMT protocol wherein Alice and Bob communicate possible multiple rounds, however, with the objective establishing a common secret-key. Using same notations as in PMT, we define a SKE protocol as follows.

\begin{definition}[SKE Protocol] \label{def-SKE}
A protocol $\Pi$ over a multipath setting is a $(k, c, \delta, \epsilon)$-SKE protocol if uses $c$ bits of communication and allows Alice and Bob to calculate key $S\in \bset^k$ and its estimate $\hat{S} \in \bset^k$, respectively, such that
\begin{eqnarray}
& \mbox{Reliability}:&~ \Pr(\hat{S} \neq S) \leq \delta, \label{reli-SKE} \\
& \mbox{Uniformity \& Secrecy}:&~ SD\left([S, View_E] ~,~ [U_k, View_E] \right) \leq \epsilon, \label{sec-SKE}
\end{eqnarray}
where $View_E$ is Eve's view of the communication and $U_k$ is an independent uniform $k$-bit string.
\\
The \emph{secret-key} rate of $\Pi$ is obtained as $R=\frac{k}{c}$. The protocol is called perfectly reliable when $\delta=0$, perfectly-secret when $\epsilon=0$; if both hold, $\Pi$ is called a perfectly (P)-SKE protocol.
\end{definition}

There is an easy way of obtaining SKE from PMT. A $(k,c,\delta,\epsilon)$-SKE protocol can be constructed by Alice generating a uniformly random key and sending it to Bob using a $(k,c,\delta,\epsilon)$-PMT protocol. The (average-case) secrecy and reliability properties of SKE follow trivially from those of PMT (in Definition \ref{def-PMT}). We conclude:
\begin{corollary}\label{corollary-SKE-PMT}
Any $(k,c,\delta,\epsilon)$-PMT protocol $\Pi$ over a setting $\Stp$ implies a $(k,c,\delta,\epsilon)$-SKE protocol $\Pi'$ over $\Stp$.
\end{corollary}

It thus suffices to show an upper bound on the secret-key rates of $(\delta,\epsilon)$ SKE protocols in the full-access case. We do this in two steps. First, we show any SKE protocol can be made perfectly-reliable by increasing secrecy parameter by $\delta$. Second, we show that relaxing secrecy causes no more than a $\frac{1}{1-O(\epsilon)}$ factor improvement to the secrecy rate. The combination gives an upper-bound on achievable rates by $(\delta,\epsilon)$ SKE (and also PMT) protocols.

\ihdr{Step 1: Considering the effect of $\delta$ in SKE}
Perfectly-reliable communication is inherent in the full-access case: Alice and Bob have access to all paths and Eve is passive. Alice and Bob may however decide to introduce intentional error to their computation in order to decrease information leakage and/or to increase the secrecy rate. In his proof of weak to strong secret-key capacity for instance, Maurer \cite{Ma93} shows that Alice can convert her almost-uniform key string to a perfectly-uniform one by passing it through a slightly noisy channel (i.e., by adding intentional error). We show below that by adding intentional error, Alice and Bob may not obtain more than trivial improvement to their leakage in the full-access case. Lemma \ref{lemma-remove-delta} shows that $\delta$-reliability can be made perfect at the price of increasing the secrecy parameter by $\delta$ and without affecting the rate.
\begin{lemma}\label{lemma-remove-delta}
For any $(k,c,\delta,\epsilon)$-SKE protocol $\Pi$ in a full-access multipath setting $\Stp$, there is a $(k, c, 0,\epsilon')$-PMT protocol $\Pi'$ over $\Stp$ with
\[\epsilon' \leq \epsilon+\delta\]
\end{lemma}
\begin{proof}
See Appendix \ref{app-proof-remove-delta}.
\end{proof}

\ihdr{Step 2: Considering the effect of $\epsilon$ in SKE}
Thanks to Step 1, we now only focus perfectly-reliable SKE. Lemma \ref{lemma-remove-epsilon} proves us an upper-bound on the secret-key rates of such protocols. The proof uses the relation between statistical-distance secrecy and mutual-information secrecy to show that allowing for $\epsilon$ leakage causes only a factor of $\frac{1}{1-O(\epsilon)}$ improvement to the rate compared to perfectly-secret SKE.

\begin{lemma}\label{lemma-remove-epsilon}
Any $(k,c,0,\epsilon)$-SKE protocol in the $(n,t_e,\len)$-full-access setting satisfies $\frac{k}{c} \leq \frac{1-\frac{t_e}{n}}{1-1.25 \epsilon- \epsilon \log \epsilon}$.
\end{lemma}
\begin{proof}
See Appendix \ref{app-remove-epsilon}.
\end{proof}

\ihdr{Combining the two steps}
Combining Lemmas \ref{lemma-remove-delta} and \ref{lemma-remove-epsilon} and applying them to Corollary \ref{corollary-SKE-PMT} gives us the following result for any $(\delta,\epsilon)$-PMT in the full-access setting.
\begin{theorem}\label{theorem-up-asymp-full}
There is no (one-way or two-way) $(k, c, \delta, \epsilon)$-PMT protocol in the $(n,t_e,\len)$-full-access setting with $\frac{k}{c} > (1-t_e/n)/(1-1.25 \epsilon'-\epsilon' \log \epsilon')$, where $\epsilon'=(\epsilon+\delta)/(1-\delta)$. This implies the $(\delta,\epsilon)$-secrecy capacity upper-bound
\[C^{\mathbf{FA}}_{\delta,\epsilon} \leq \frac{1-t_e/n}{1-1.25 \epsilon'-\epsilon' \log \epsilon'},\]
and the AP-secrecy capacity bounds (using the lower-bound in (\ref{bounds-C-FA-0}))
\begin{eqnarray}\label{bounds-C-FA-ap}
1-\frac{t_e}{n}-\Delta \leq C^{\mathbf{FA}}_{\sim 0} \leq 1-\frac{t_e}{n}.
\end{eqnarray}
\end{theorem}

The bounds (\ref{bounds-C-FA-0}) and (\ref{bounds-C-FA-ap}) show that both secrecy capacities fall in the same range. Assuming that $\Delta$ is negligible, we reach that P-secrecy and AP-secrecy capacities both equal $1-\frac{t_e}{\tmin}$. With this assumption, we have the conclusion:

\begin{center}
{\it Relaxing security requirements from perfect to asymptotically-perfect does NOT help improve the secrecy rate in the ``full-access'' case.}
\end{center}

\section{PMT in the general multipath setting}
Although the relaxation of security to asymptotically perfect did not cause rate improvement to full-access PMT, it may be of benefit to the general case where $t_a,t_b \leq n$. This motivates us to extend our study to the general multipath setting. We start by considering P-PMT protocols and next look at AP-PMT protocols.

\subsection{P-PMT: capacity and construction}
We show that the P-secrecy capacity of any $(n,t_a,t_b,t_e,\len)$ multipath communication setting equals approximately $C_0 \approx [1-\frac{t_e}{\tmin}]_+$, where $\tmin=\min(t_a,t_b)$. For this we shall derive lower and upper bounds that are almost tight. The lower bound follows trivially by applying one-round P-PMT schemes $\FF^{pol}_0$ and $\FF^{alg}_0$ (for the full-access case) to a fixed set of $\tmin$ paths and forgetting about the rest $n-\tmin$ paths. But this protocol is simple and one-round; one may wonder if there are (possibly) multiple-round P-PMT protocols that achieve better rates. We prove an upper-bound on the secrecy rates of (one-way and interactive) P-PMT protocols, which shows there is no room for improving the P-secrecy rate of $1-\frac{t_e}{\tmin}$ in the general (one-way/two-way) multipath setting. The results also conclude the impossibility of P-PMT (of any rate) when $t_e \geq \tmin$.

\bhdr{Lower-bound via one-round P-PMT}
The lower-bound on $C_0$ is attained by using the same schemes $\FF^{pol}_0$ and $\FF^{alg}_0$ (selectively based on whether $2n-t < 2^\len$), however, over a fixed (hard-coded) set of $\tmin$ paths. By fixing these paths, Alice and Bob can realize a $(\tmin,t_e,\len)$-full-access setting for which the above P-PMT constructions promise the secrecy rate of $1-\frac{t_e}{\tmin} -(2^{\frac{\len}{2}-2}-0.25)^{-1}$ (see (\ref{bounds-C-FA-0})). This of course requires that $t_e \leq \tmin$.

\bhdr{Upper-bound on all P-PMT rates}
We derive an upper-bound on the secrecy rate of any (possibly multiple-round) P-PMT protocol that can be designed for a general multipath setting. Informally, the upper-bound is proved by showing that to provide ``perfect'' secrecy, Alice and Bob should be prepared for the ``worst'' case when Eve captures $t_e$ of their $\tmin$ common communication paths, even if such a worst case occurs with a very small probability. This is similar to a setting where Alice and Bob have full access to a total of $\tmin$ (public) paths, $t_e$ of which can be accessed by Eve; hence the maximum rate $1-t_e/\tmin$.

\begin{lemma}\label{lemma-upperbound-0}
There is no (possibly multiple-round) $(k,c,0,0)$-PMT protocol over the $(n,t_a,t_b,t_e,\len)$-multipath setting with rate $R=\frac{k}{c} > [1-\frac{t_e}{\tmin}]_+$.
\end{lemma}
\begin{proof}
See Appendix \ref{app-proof-upperbound-0}.
\end{proof}
The upper-bound implies that the P-PMT Scheme based on $\FF^{pol}_0$ or $\FF^{alg}_0$ is nearly optimal. Theorem \ref{theorem-0secrecy} summarizes the results on the P-secrecy capacity in the general setting.

\begin{theorem}\label{theorem-0secrecy}
The P-secrecy capacity of any $(n,t_a,t_b,t_e,\len)$ multipath setting is bounded as $L_0 \leq C_{0} \leq U_0$, where
\begin{eqnarray}
L_0 \stackrel{\triangle}{=} [1-\frac{t_e}{\tmin}-(2^{\frac{\len}{2}-2}-0.25)^{-1}]_+, ~~\mbox{and}~~ U_0 \stackrel{\triangle}{=} [1-\frac{t_e}{\tmin}]_+.
\end{eqnarray}
and the lower-bound is achieved by an explicit one-round PMT protocol.
\end{theorem}

\subsection{AP-PMT: Capacity and constructions}
We showed that in the full-access case, allowing for arbitrarily small $\delta,\epsilon>0$ leads to the same PMT rate upper-bound as if $\delta=\epsilon=0$ (see Section \ref{subsec-AP-full}). The result however cannot be applied to the general multipath setting: The upper-bound proof of Theorem \ref{theorem-up-asymp-full} does not apply to the $(\delta,\epsilon)$-secrecy capacity in the partial-access case. This is because Lemma \ref{lemma-remove-delta} only works for the full-access case where reliable communication is inherent. This leaves us with the trivial upper-bound of
\begin{eqnarray}\label{upperbound-C-ap0}
C_{\sim 0} \leq U_{\sim 0} \stackrel{\triangle}{=} 1-\frac{t_e}{n}.
\end{eqnarray}

At first look, the above upper seems far from being tight. It seems impossible for Alice and Bob to reach secrecy rates up to $1-\frac{t_e}{n}$, which is regardless of their access capabilities $t_a$ and $t_b$.We prove provided that $\len$ is sufficiently large, the upper bound is almost tight and there are AP-PMT families which can get close to this rate even if $t_a$ and $t_b$ are small. For one-way multipath setting, the required connectivity condition is $t_b > t_e$; for two-way setting however, AP-PMT is always possible only if $t_e < n$ and $t_a,t_b >0$. We introduce three AP-PMT schemes, two of which are one-round constructions for the one-way multipath setting and the last is a two-round construction for the two-way setting.

\subsubsection{AP-PMT Approach}
All three AP-PMT schemes, we introduce below, consist of two primitive blocks: (i) \emph{low-rate key establishment block} and (ii) \emph{high-rate coordinated PMT block}. Both blocks take use of the algebraic-geometric quasi-ramp SSS of Section \ref{subsec-SSS}. The key-establishment block allows Alice and Bob to share a long enough secret-key $W$; this can be via either \emph{one-round key transport} or \emph{interactive key agreement} depending on the setting privileges. By the coordinated PMT block, Alice sends her message over secret paths chosen based on $W$: Since Eve is unaware of $W$, the coordinated PMT rate can be as high as (almost) $1-\frac{t_e}{n}$, as if Alice and Bob possess full access. Provided that the communication overhead of the key-establishment block is negligible, the overall rate of the scheme tends toward that of the second block, i.e., $1-\frac{t_e}{n}$. The size of the coordination key $W$ depends on how long message is to be transmitted. We use parameters $q_1, q_2 \in \Nds$ for the number of time-intervals for key establishment and coordinated PMT, respectively; hence, to allow for flexible message length transmission. Each interval of coordinated PMT requires $w=\lceil \log {n \choose \tmin} \rceil$ bits of shared key for Alice and Bob to choose secret paths; on the other hand, each key-establishment interval produces some $r_1 \len$ (defined in the sequel) bits of key. This clarifies the relation $q_1 r_1 \len = q_2 w$ between the number of intervals. In the following, we provide detailed descriptions for the AP-PMT schemes separately.

\subsubsection{One-round AP-PMT for $t_e \leq \tmin$}
We introduce a one-round AP-PMT scheme $\FF_1$ with an AP-secrecy rate close to $1-\frac{t_e}{n}$. The scheme has perfect reliability, but allows for negligible leakage. It composes a key-transport block (made by P-PMT) and a coordinated PMT block as follows. Given the $(n,t_a,t_b,t_e,\len)$ multipath setting, define $w=\lceil \log{n~ \choose \tmin} \rceil$. For arbitrarily small $\psi>0$, and sufficiently large $q_1 \in \Nds$ (to be defined in Theorem \ref{theorem-FF1}), define
\footnote{Here, we assume that $q_1$ is chosen such that $w$ divides $r_1 \len$.}
\begin{eqnarray}
&& g_1= \lceil \frac{q_1 (2 \tmin -t_e)}{2^{\len/2}-1}\rceil,~~~ r_1= q_1(\tmin -t_e) - 2 g_1 \nonumber \\
&& q_2 = \frac{r_1 \len}{w},~~~ t'_e=(1+\psi)\frac{\tmin t_e}{n},~~~ g_2 =\lceil \frac{q_2 (2 \tmin-t'_e)}{2^{\len/2}-1} \rceil,~~~ r_2= q_2 (\tmin-t'_e)-2g_2. \label{tk-FF1}
\end{eqnarray}

Let $(\share_{alg,1},\rec_{alg,1})$ denote the algebraic-geometric $(q_1 \tmin - 2g_1, r_1, g_1, q_1 \tmin)$-quasi-ramp SSS over  $\Fds_{2^\len}$ used for key transport, and $(\share_{alg,2},\rec_{alg,2})$ denote the algebraic-geometric $(q_2 \tmin - 2g_2, r_2, g_2, q_2 \tmin)$-quasi-ramp SSS over  $\Fds_{2^\len}$ used for coordinated PMT. Let ${\cal T}_0$ be a set of fixed $\tmin$ paths publicly known to the protocol. Let $S \in \Fds_{2^\len}^{r_2}$ be the message to be transmitted by Alice.

\ihdr{One-round $(0,\epsilon)$-PMT scheme $\FF_1$}
\begin{enumerate}[(i)]
\item \emph{Key transport (intervals $1$ to $q_1$).} Alice generates random keys $W = (W_1,\dots, W_{q_2}) \in \bset^{q_2 w}$; note that $q_2 w = r_1 \len$. She obtains $q_1 \tmin$ shares of $W$ as $\left( X_{i,j}\right)_{1\leq i \leq q_1, 1\leq j \leq \tmin} = \share_{alg,1}(W)$ and sends the $\tmin$ shares $\left(X_{i,j}\right)_{1\leq j \leq \tmin}$ over distinct paths in ${\cal T}_0$ during interval $1\leq i\leq q_1$. Having received all shares, Bob reconstructs the transported key as $W=\rec_{alg,1}(X)$. The number of bits communicated at this stage equals:
    \begin{eqnarray}
        c_1= q_1 \tmin \len.  \label{c1-FF1}
    \end{eqnarray}
\item \emph{Coordinated PMT (intervals $q_1+1$ to $q_1+q_2$).} Alice and Bob use key element $W_i \in \bset^w$ to agree on a path set ${\cal T}_i$ of size $\tmin$ for interval $q_1+i$. This gives them a total of $q_2 \tmin$ paths available for $q_2$ consecutive time-intervals. Alice calculates shares of her message as $\underline{Y}=\left( Y_{i,j}\right)_{1\leq i \leq q_2, 1\leq j \leq \tmin} = \share_{alg,2}(S)$ and sends the $\tmin$ shares $\left(Y_{i,j}\right)_{1\leq j \leq \tmin}$ over ${\cal T}_i$ during interval $q_1+i$. Having received all shares $\underline{Y}$, Bob calculates $S=\rec_{alg,2}(\underline{Y})$. This stage requires communicating the following number of bits:
    \begin{eqnarray}
        c_2= q_2 \tmin\len. \label{c2-FF1}
    \end{eqnarray}
\end{enumerate}

The secrecy rate of block (ii) equals $\frac{r_2 \len}{c_2} \approx 1 - \frac{t_e'}{\tmin}$ which tends to $1 - \frac{t_e}{n}$ for arbitrarily small $\psi>0$. The overall secrecy rate of $\FF_1$ is lower due to the communication overhead of block (i). Theorem \ref{theorem-FF1} shows the achievable rates by $\FF_1$ as a function of setting parameters.

\begin{theorem}\label{theorem-FF1}
For any small $\psi,\epsilon>0$, the scheme $\FF_1$ gives $(0,\epsilon)$-PMT and AP-PMT families over an $(n,t_a,t_b,t_e,\len)$ multipath setting with $t_e <\tmin \leq n$. The AP-secrecy rate of the scheme equals
\[R_{\FF_1:\sim 0} =\frac{1-\frac{t_e}{n}-\Delta}{1+\xi_1}, ~~\mbox{where}~~ \xi_1=\frac{\log(\frac{en}{\tmin})}{\len (1-\frac{t_e}{\tmin}-\Delta)} ~~\mbox{and}~~ \Delta=(2^{\frac{\len}{2}-2}-0.25)^{-1}.\]
\end{theorem}
\begin{proof}
See Appendix \ref{app-proof-FF1}.
\end{proof}

\begin{remark}
It is crucial to use an algebraic-geometric SSS, particularly for the second block, coordinated PMT. Expecting arbitrarily small $\epsilon>0$ requires using SSS with sufficiently many ($q_2 \tmin$) shares for constant field size $2^\len$.
\end{remark}

The secrecy rate $R_{\FF_1:\sim 0}$ is close to the upper-bound expression (\ref{upperbound-C-ap0}), except there is a subtracting factor $\Delta$ due to using quasi-threshold SSS and a divisive factor $(1+\xi_1)$ because of the communication overhead $\xi_1$ of the key-transport block. Comparing this rate with the P-secrecy capacity shows an improvement of secrecy rates by using AP-PMT.

\subsubsection{One-round AP-PMT for $t_e\geq \tmin$}
The scheme $\FF_1$ cannot achieve any positive secrecy rate when $t_e\geq \tmin$; hence, the AP-secrecy capacity for this range is unknown. We observe the following two limitations of $\FF_1$.
\begin{itemize}
\item \emph{$\FF_1$ has perfect reliability.} Allowing for positive but small failure probability $\delta>0$ in transmission may lead to higher secrecy rates.
\item \emph{$\FF_1$ is non-interactive.} It may be possible to improve the rate via interactive communication between Alice and Bob.
\end{itemize}

The above two reasons encourage us to investigate families of AP-PMT schemes that allow non-perfect reliability and/or are interactive. We first consider only one-way communication and study whether relaxing perfect reliability can make AP-PMT possible even when $t_e\geq \tmin$.

Below, we describe the PMT scheme $\FF_2$ which relaxes both reliability and secrecy of transmission, but results in PMT possibility in a wider connectivity range of $t_e<t_b$. Here, the protocol fixes a set ${\cal T}_0$ of $\max(t_a,t_b) \leq n' \leq n$ (instead of $\tmin$) paths for key transport, i.e., Alice sends over $t_a$ paths and Bob listens over $t_b$ paths, both of which are selected independently at random (if applicable) from the $n'$ fixed paths. The reliability is not perfect since the path selection can be random. The details of the scheme are given below.

Given the $(n,t_a,t_b,t_e,\len)$ multipath setting, let $w=\lceil \log{n~ \choose \tmin} \rceil$. For arbitrarily small $\psi>0$, and sufficiently large $q_1 \in \Nds$ (to be defined in Theorem \ref{theorem-FF2}), define
\footnote{Here, we assume that $q_1$ is chosen such that $w$ divides $r_1 \len$.}
\begin{eqnarray}
&& t'_{b,1}= (1-\psi)\frac{t_a t_b}{n'}, ~~~ t'_{e,1}=(1+\psi)\frac{t_a t_e}{n'}, ~~~ g_1= \lceil \frac{q_1 (t_a+t'_b-t_e')}{2^{\len/2}-1}\rceil,~~~ r_1= q_1(t'_{b,1}-t'_{e,1}) - 2 g_1 \nonumber \\
&& q_2 = \frac{r_1 \len}{w},~~~ t'_{e,2}=(1+\psi)\frac{\tmin t_e}{n},~~~ g_2 =\lceil \frac{q_2 (2 \tmin-t'_{e,2})}{2^{\len/2}-1} \rceil, ~~~ r_2= q_2(\tmin-t'_{e,2}) - 2 g_2. \label{tk-FF2}
\end{eqnarray}
Let $(\share_{alg,1},\rec_{alg,1})$ denote the algebraic-geometric $(q_1 t'_{b,1} - 2g_1, r_1, g_1, q_1 t_a)$-quasi-ramp SSS over field $\Fds_{2^\len}$ used for key transport, and $(\share_{alg,2},\rec_{alg,2})$ denote the algebraic-geometric $(q_2 \tmin - 2g_2, r_2, g_2, q_2 \tmin)$-quasi-ramp SSS over field $\Fds_{2^\len}$ used for coordinated PMT. Let $S \in \Fds_{2^\len}^{r_2}$ be the message to be transmitted by Alice.

\ihdr{One-round $(\delta,\epsilon)$-PMT scheme $\FF_2$}
\begin{enumerate}[(i)]
\item \emph{Key transport (intervals $1$ to $q_1$).} Alice generates random keys $W = (W_1,\dots, W_{q_2}) \in \bset^w$ and obtains $q_1 t_a$ shares of $W$ as $\left( X_{i,j}\right)_{1\leq i \leq q_1, 1\leq j \leq t_a} = \share_{alg,1}(W)$. In each round $1\leq i\leq q_1$, she sends the $t_a$ shares $\left(X_{i,j}\right)_{1\leq j \leq t_a}$ over $t_a$ (possibly) random paths from ${\cal T}_0$, and Bob listens over $t_b$ (possibly) random paths from ${\cal T}_0$. If Bob's overall observation $\underline{X}'$ includes less than $q_1 t'_{b,1}$ share elements ($X_{i,j}$'s), he aborts and chooses a random $\hat{S} \in \Fds_{2^\len}^{r_2}$; otherwise, he reconstructs the transported key as $W=\rec_{alg,1}(\underline{X}')$. The number of bits communicated at this stage equals:
    \begin{eqnarray}
        c_1 = q_1 t_a \len.  \label{c1-FF2}
    \end{eqnarray}
\item \emph{Coordinated PMT (intervals $q_1+1$ to $q_1+q_2$).} Alice and Bob use key element $W_i \in \bset^w$ to agree on a path set ${\cal T}_i$ of size $\tmin$ for interval $i+q_1$. This gives them a total of $q_2 \tmin$ secret paths available for $q$ consecutive time-intervals. Alice calculates shares of message as $\underline{Y}=\left( Y_{i,j}\right)_{1\leq i \leq q_2, 1\leq j \leq \tmin} = \share_{alg,2}(S)$ and sends the $\tmin$ shares $\left(Y_{i,j}\right)_{1\leq j \leq \tmin}$ over ${\cal T}_i$ during interval $i+q_1$. Having received all shares $\underline{Y}$, Bob calculates $S=\rec_{alg,2}(\underline{Y})$. At this stage, the number of communicated bits is:
    \begin{eqnarray}
        c_2= q_2 \tmin \len. \label{c2-FF2}
    \end{eqnarray}
\end{enumerate}

Like $\FF_1$, the coordinated PMT rate in the above scheme tends to $1-\frac{t_e}{n}$, whereas the overall rate is lower due to the key-transport overhead. The difference with $\FF_1$ is that $\FF_2$ does not use perfectly-reliable key transport; this allows for pushing the multipath connectivity condition to $t_e>t_b$. The scheme relies on the following assumptions: In key transport, Bob receives at least $q_1 t'_{b,1}$ shares and Eve receives at most $q_1 t'_{e,1}$ shares of $W$; furthermore, in coordinated PMT Eve receives at most $q_2 t'_{e,2}$ shares of $S$. If the above assumptions hold, the PMT becomes perfectly reliable and secret thanks for the secrecy and reconstruction properties of the two quasi-ramp SS schemes. The $\delta$-reliability and $\epsilon$-secrecy properties come from the case that the assumptions fail.
\begin{theorem}\label{theorem-FF2}
For any small $\psi, \delta, \epsilon>0$, the scheme $\FF_2$ gives $(\delta,\epsilon)$-PMT and AP-PMT families over any $(n,t_a,t_b,t_e,\len)$ multipath setting with $t_e<t_b$. The AP-secrecy rate of this scheme reaches
\[R_{\FF_2:\sim 0} = \frac{1-\frac{t_e}{n}-\Delta}{1+\xi_2}, ~~\mbox{where}~~ \xi_2 = \frac{\log(\frac{e n}{\tmin})}{\len \left(\frac{t_b-t_e}{n'} - \Delta \right)}, ~~ n'=\max(t_a,t_b), ~~\mbox{and}~~ \Delta=(2^{\frac{\len}{2}-2}-0.25)^{-1}.\]
\end{theorem}
\begin{proof}
See Appendix \ref{app-proof-FF2}.
\end{proof}

\begin{remark}\label{remark-FF2}
The Scheme $\FF_2$ can be simplified to achieve a higher rate when $t_b=n$. For this special case, only Stage (i) the key transport block can be used to serve message transmission at rate $\frac{t_b-t_e}{t_b}-\Delta = 1- \frac{t_e}{n} - \Delta$, by choosing $n'=n$.
\end{remark}

\bhdr{Impossibility of one-way PMT for $t_e\geq t_b$}
It is impossible to obtain AP-PMT in one-round when $t_e \geq t_b$, i.e., when Eve has access to more paths than Bob. The intuition is that Eve will have total advantage over Bob and any protocol that lets Bob obtain a good estimate of the message will let Eve too.
\begin{proposition}\label{prop-ow-t_e>mb}
There is no one-round $(k,c, \delta,\epsilon)$-PMT protocol of rate $R = \frac{k}{c} > \frac{2\epsilon}{1-\delta-\alpha}$ to transmit $k\geq 3/\alpha$ bits of messages over a multipath setting with $t_e\geq t_b$, implying the AP-secrecy capacity of $C_{\sim 0}=0$ for $t_e\geq t_b$.
\end{proposition}

\begin{proof}
See Appendix \ref{app-proof-ow-t_e>mb}.
\end{proof}

\bhdr{Implication to one-way capacity}
Putting all things together, we reach the following conclusion about the AP-secrecy capacity for one-way communication. When $t_e \geq t_b$, one-way AP-PMT is impossible and the capacity is zero. When $\tmin \leq t_e < t_b$ (if applicable), only Scheme $\FF_2$ can be used and hence the capacity is lower-bounded by $R_{\FF_2:\sim 0}$. Finally, when $t_e < \tmin$, both schemes $\FF_1$ and $\FF_2$ can be used and the lower bound on the capacity equals the maximum of the two rates $R_{\FF_1:\sim 0}$ and $R_{\FF_2:\sim 0}$.

\begin{corollary}\label{corollary-oneway-AP-cap}
The AP-secrecy capacity of any one-way $(n,t_a,t_b,t_e,\len)$-multipath setting is bounded as $\overrightarrow{L}_{\sim 0} \leq \overrightarrow{C}_{\sim 0} \leq \overrightarrow{U}_{\sim 0}$, where
\begin{eqnarray}
\overrightarrow{L}_{\sim 0} =
\begin{cases}
[\frac{1-\frac{t_e}{n}-\Delta}{1+\min(\xi_1,\xi_2)}]_+, & \mbox{if}~~ t_e <  \tmin \\
[\frac{1-\frac{t_e}{n}-\Delta}{1+\xi_2}]_+, & \mbox{if}~~ \tmin \leq t_e < t_b \\
0, & \mbox{if}~~ t_e\geq t_b
\end{cases}
, ~~\mbox{and}~~
\overrightarrow{U}_{\sim 0} =
\begin{cases}
1-\frac{t_e}{n}, & \mbox{if}~~ t_e < t_b \\
0, & \mbox{if}~~ t_e\geq t_b
\end{cases}.
\end{eqnarray}
\end{corollary}

In the next section, we show that using interactive communication, Alice and Bob attain AP-PMT with positive rate even if $t_e\geq t_b$.

\subsection{AP-PMT: Always positive rates via two-way communication}
We introduce a two-round (interactive) AP-PMT scheme and show its security even under the condition that $t_e\geq t_b$. The idea is similar to the previous schemes except that the first block is now an \emph{interactive key-agreement} (rather than key-transport) protocol: Bob sends random data elements over random paths and Alice publicly responds (over a fixed path) which elements she has received. Having known about their common elements, Alice and Bob apply privacy amplification to convert them into a secret-key. We take use of the algebraic-geometric SSS for  privacy amplification.

Given the $(n,t_a,t_b,t_e,\len)$ multipath setting, let $w_1=\lceil \log{n~ \choose t'_{a,1}} \rceil$ and $w_2=\lceil \log{n~ \choose \tmin} \rceil$. For arbitrarily small $\psi>0$, and sufficiently large $q_1 \in \Nds$ (to be defined in Theorem \ref{theorem-FF3}), define
\footnote{Here, we assume that $q_1$ is chosen such that $w_2$ divides $r_1 \len$.}
\begin{eqnarray}
&& t'_{a,1}= (1-\psi)\frac{t_a t_b}{n}, ~~~ t'_{e,1}=(1+\psi)\frac{t'_{a,1} t_e}{n}, ~~~ g_1= \lceil \frac{q_1 (2t'_{a,1}-t'_{e,1})}{2^{\len/2}-1}\rceil,~~~ r_1= q_1(t'_{a,1}-t'_{e,1}) - 2 g_1 \nonumber \\
&& q_2 = \frac{r_1 \len}{w_2},~~~ t'_{e,2}=(1+\psi)\frac{\tmin t_e}{n},~~~ g_2 =\lceil \frac{q_2 (2 \tmin-t'_{e,2})}{2^{\len/2}-1} \rceil, ~~~ r_2= q_2(\tmin-t'_{e,2}) - 2 g_2. \label{tk-FF3}
\end{eqnarray}
Let $(\share_{alg,1},\rec_{alg,1})$ denote the algebraic-geometric $(q_1 t'_{a,1} - 2g_1, r_1, g_1, q_1 t'_{a,1})$-quasi-ramp SSS over field $\Fds_{2^\len}$ used for key transport, and $(\share_{alg,2},\rec_{alg,2})$ denote the algebraic-geometric $(q_2 \tmin - 2g_2, r_2, g_2, q_2 \tmin)$-quasi-ramp SSS over field $\Fds_{2^\len}$ used for coordinated PMT. Let $S \in \Fds_{2^\len}^{r_2}$ be the message to be transmitted by Alice.

\ihdr{Two-round $(\delta,\epsilon)$-PMT scheme $\FF_3$}
\begin{enumerate}[(i)]
\item \emph{Interactive key agreement (intervals $1$ to $q_1$).} Bob generates $q_1 t_b$ random elements $\underline{X} = (X_{i,j})_{1\leq i\leq q_1, 1\leq j\leq t_b} \in (\Fds_2^\len)^{q_1 t_b}$ and sends them in $q_1$ time intervals: In each time interval $1\leq i\leq q$, he sends $(X_{i,j})_{1\leq j \leq t_b}$ over $t_b$ (independently) randomly chosen paths. If Alice's observation includes less than $q_1 t'_{a,1}$ field elements from $X$ (sent by Bob), she sends an abort signal and Bob outputs a random message estimation $\hat{S}\in_R \Fds_{2^\len}^{r_2}$. Otherwise, let $X_A\subseteq X$ consist of the first $q_1 t'_{a,1}$ field elements observed by Alice over paths $({\cal P}_i)_{1\leq i\leq q_1}$. Alice sends the information of $({\cal P}_i)_{1\leq i\leq q_1}$ (requiring at most $q_1 w_1$ bits) over a (fixed/public) path to Bob. Alice and Bob both know $X_A$ and use them as shares of quasi-ramp SSS to obtain the key $W=(W_1,\dots,W_{q_2}) = \rec_{alg,1}(\tilde{X}_A) \in \Fds_{2^\len}^{r_1}$ (note that $r_1 \len = q_2 w_2$). The communication overhead of this stage equals:
    \begin{eqnarray}\label{c1-FF3}
        c_1 = q_1 t_b \len + q_1 w_1.
    \end{eqnarray}
\item \emph{Coordinated PMT (intervals $q_1+1$ to $q_1+q_2$).} Alice and Bob use key elements $W_i \in \bset^{w_2}$ to agree on $q_2$ path sets ${\cal T}_i$ of size $\tmin$ for intervals $q_1+1 \leq q_1+i \leq q_1+q_2$. Alice calculates shares of message as $\underline{Y}=\left( Y_{i,j}\right)_{1\leq i \leq q_2, 1\leq j \leq n} = \share_{alg,2}(S)$ and sends $\left(Y_{i,j}\right)_{1\leq j \leq \tmin}$ over ${\cal T}_i$ during interval $q_1+i$. Having received $\underline{Y}$, Bob calculates $\hat{S}=S=\rec_{alg,2}(\underline{Y})$. The number of communicated bits is:
    \begin{eqnarray}\label{c2-FF3}
        c_2= q_2 \tmin \len.
    \end{eqnarray}
\end{enumerate}

\begin{theorem}\label{theorem-FF3}
For any small $\psi,\delta,\epsilon>0$, the scheme $\FF_3$ gives $(\delta,\epsilon)$-PMT and AP-PMT families over any $(n,t_a,t_b,t_e,\len)$ multipath setting, with $t_a,t_b>0$ and $t_e < n$. The AP-secrecy rate of this family equals:

\begin{eqnarray*}
 R_{\FF_3:\sim 0} = \frac{1-\frac{t_e}{n} - \Delta}{1+ \xi_3}, ~~~ \mbox{where}~~~ \xi_3 = \frac{\left( \frac{n}{t_a} + \frac{\log(e n^2 /(t_a t_b) } {\len} \right) \log\frac{e n}{\tmin}}{\len \left(1 - \frac{t_e}{n} - \Delta \right)} ~~\mbox{and}~~ \Delta=(2^{\frac{\len}{2}-2}-0.25)^{-1}.
\end{eqnarray*}
\end{theorem}
\begin{proof}
See Appendix \ref{app-proof-FF3}.
\end{proof}

\bhdr{Implication to two-way capacity} When two-way communication is allowed, the AP-secrecy capacity is trivially upper-bound by $1-\frac{t_e}{n}$, unless obviously when $t_a=0$ or $t_b=0$. We have the following lower bounds on the other hand. For $t_e < \tmin$, all schemes $\FF_1$, $\FF_2$, and $\FF_3$ can be used and the maximum rate shows the lower bound. When $\tmin \leq t_e  < t_b$ (if applicable), the lower bound is achieved by $\FF_2$ or $\FF_3$. Finally when $t_e \geq t_b$, onlt the interactive scheme $\FF_3$ is usable and the capacity is lower-bounded by $R_{\FF_3:\sim 0}$.

\begin{corollary}\label{corollary-twoway-AP-cap}
The AP-secrecy capacity of any $(n,t_a,t_b,t_e,\len)$-multipath setting satisfies $L_{\sim 0} \leq C_{\sim 0} \leq U_{\sim 0}$, where
\begin{eqnarray}
L_{\sim 0} =
\begin{cases}
[\frac{1-\frac{t_e}{n}-\Delta}{1+\min(\xi_1,\xi_2,\xi_3)}]_+, & \mbox{if}~~ t_e <  \tmin \\
[\frac{1-\frac{t_e}{n}-\Delta}{1+\min(\xi_2,\xi_3)}]_+, & \mbox{if}~~ 0 <\tmin \leq t_e < t_b \\
[\frac{1-\frac{t_e}{n} - \Delta}{1+ \xi_3}]_+,  & \mbox{if}~~ 0 <t_b< t_e ~\wedge~ t_a>0 \\
0,      & \mbox{else}
\end{cases},
~\mbox{and}~
U_{\sim 0} = \begin{cases}
1-\frac{t_e}{n}, & \mbox{if}~~  t_a,t_b >0\\
0,  & \mbox{else}
\end{cases}.~~
\end{eqnarray}
\end{corollary}

\subsection{Comparison of P-secrecy and AP-secrecy rates}\label{subsec-rateCompare}
We have proved that in the case of partial-access multipath communication, Alice and Bob can achieve higher secrecy rates if they choose AP-PMT protocols over P-PMT ones. We also introduced AP-PMT families with rates close to the upper-bound $1-\frac{t}{n}$. It has remained unclear, however, how much of rate improvement is attained by AP-PMT protocols in practice and how close the resulting rate is to the upper-bound. We address this by analyzing P-secrecy and AP-secrecy capacities for typical multipath parameters that match practical communication scenarios.

\subsubsection{P-secrecy vs. AP-secrecy capacities}
Although the two capacities equal for the full-access case ($t_a=t_b=n$), we argue that when $\tmin < n$, the AP-secrecy capacity is strictly larger that the P-secrecy capacity for almost all cases. The following argument is given based on Corollaries \ref{corollary-oneway-AP-cap} and \ref{corollary-twoway-AP-cap}.

\ihdr{Case 1: $t_e \geq \tmin$}
The P-secrecy capacity equals zero and is thus strictly less than the AP-secrecy capacity provided that $1-\frac{t_e}{n} > \Delta$, which is true for the connectivity range of $\tmin \leq t_e < n (1-\Delta)$.

\ihdr{Case 2: $t_e < \tmin$}
In this case, both secrecy capacities are positive. However, achievable AP-secrecy rates are strictly higher than the P-secrecy capacity if $\frac{1-t_e/n-\Delta}{1+\xi_1} > 1-\frac{t_e}{\tmin}$, which implies the range $\frac{\xi_1+\Delta}{\frac{1+\xi_1}{\tmin}-\frac{1}{n}} < t_e < \tmin$.

\begin{corollary}\label{corollary-pa-a-p}
The strict inequality $C_{\sim 0} > C_{0}$ holds for any $(n,t_a,t_b,t_e,\len)$-multipath setting that satisfies
\[\frac{\xi_1+\Delta}{\frac{1+\xi_1}{\tmin}-\frac{1}{n}} < t_e < (1-\Delta) n,\]
where $\Delta$ and $\xi_1$ are given in Theorem \ref{theorem-FF1}. As $\Delta$ goes to zero, the above connectivity range tends to
\[\frac{\alpha \log(\frac{e}{\alpha})}{\len (1-\alpha)}n < t_e < n, ~~\mbox{where}~~ \alpha=\frac{\tmin}{n}.\]
\end{corollary}

Corollary \ref{corollary-pa-a-p} clearly does not imply any range in the full-access case where $\alpha=1$; however, with a slight deviation from full access, the superiority of AP-PMT rates holds for a wide connectivity range. Assuming for example $\len = 100$ (hence $\Delta < 10^{-14} \approx 0$) and partial access $\tmin = 0.2 n$, AP-PMT shows strictly higher rates for the wide range of $0.01 n < t_e < n$. This however should not imply that our constructions achieve the AP-secrecy capacity for this range. The next section shows how deviation from full-access results in a larger gap between the lower and upper bounds on $C_{\sim 0}$.

\subsubsection{Numerical analysis of secrecy capacity}
Figure \ref{fig-cap_te} graphs the lower and upper bounds on $C_{\sim 0}$ as well as $C_{0}$ for different values of $\beta=\frac{t_e}{n}$, assuming $\len=100$ (hence $\Delta \approx 0$) and $t_a = t_b = 0.2 n$. For this value of $\len$, we approximate $C_0 \approx 1-\frac{t_e}{\tmin} = 1-5 \beta$ (see Theorem \ref{theorem-0secrecy}). The capacity $C_0$ is shown by a solid line and the bounds $\overline{L}_{\sim 0}$ and $\overline{U}_{\sim 0}$ are shown by dotted and dashed lines, respectively. The graph clearly illustrates the benefit of using AP-PMT over P-PMT in the multipath setting. Both P-secrecy and AP-secrecy capacities show linear decrement with respect to $\beta$; however, $C_0$ drops much faster and equals 0 for $\beta \geq 0.2$. The lower bound on $C_{\sim 0}$ shows that this capacity remains positive and close to the upper-bound $U_{\sim 0} = 1-\frac{t_e}{n}$ throughout. For $\beta \leq 0.15$, the lower bound is very close to the upper bound and is achieved by one-round AP-PMT ($\FF_1$ or $\FF_2$). Outside of this range, the achievable rates by our one-round AP-PMT rate drops drastically and tend to 0 at $\beta =0.2$. This is not surprising since one-way AP-PMT is impossible when $t_e \geq t_b$ (implying $\beta \geq 0.2$). For $\beta \geq 0.15$, the lower-bound is attained by our two-round scheme $\FF_3$. Observe that the gap between the two bounds for this range which is due to the overheard factor $\xi_3$ (see rate $R_{\FF_3:\sim 0}$ in Theorem \ref{theorem-FF3}).

\begin{figure}[hbt]
  \centering
  \subfigure[w.r.t. $\beta = \frac{t_e}{n}$ for $t_a = t_b = 0.2 n$]{
    \includegraphics[width=.45\textwidth]{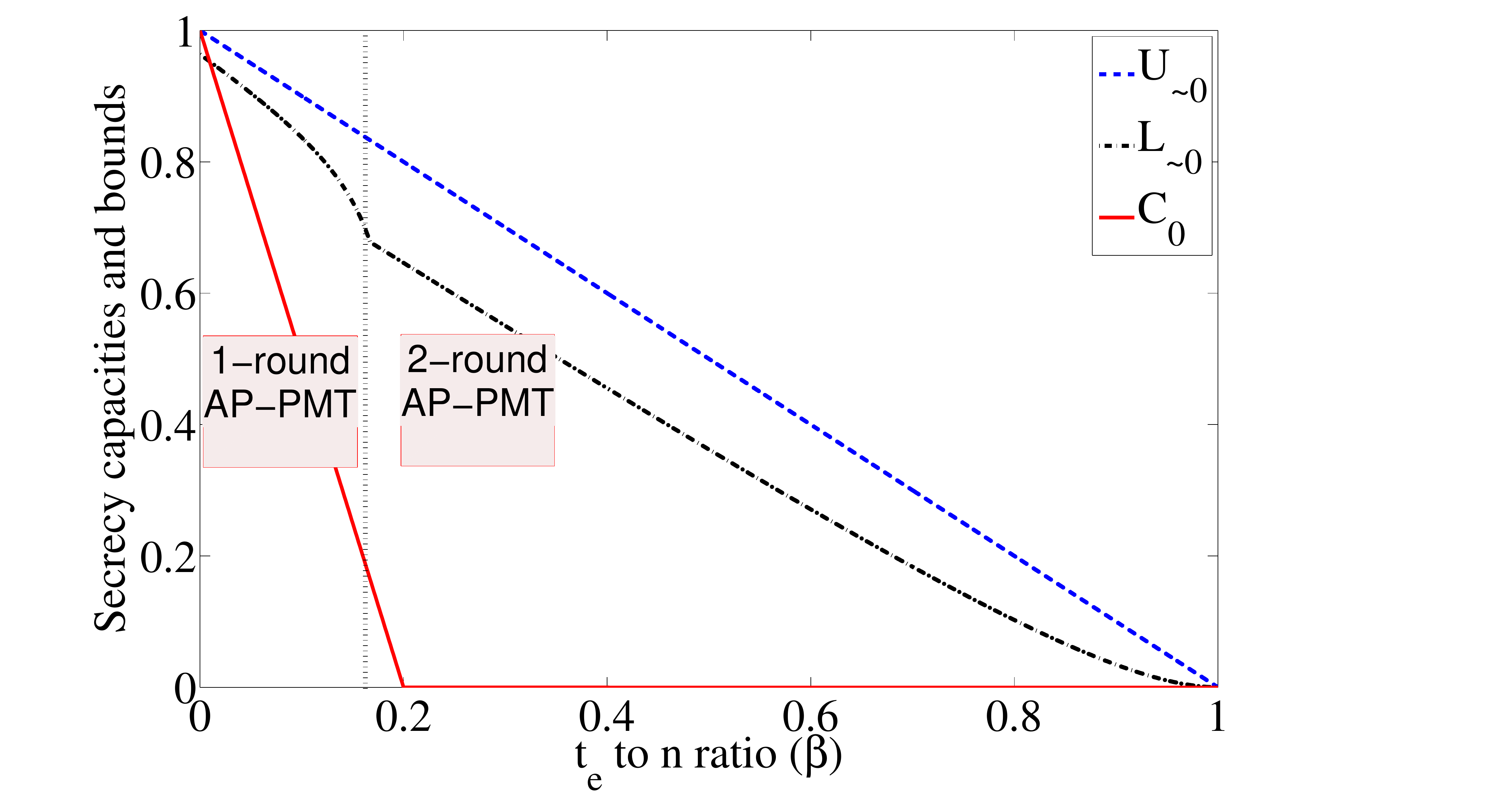}
    \label{fig-cap_te}
  }
  \subfigure[w.r.t. $\alpha = \frac{\tmin}{n}$ for $t_a = t_b$ and $t_e = 0.2 n$]{
    \includegraphics[width=.45\textwidth]{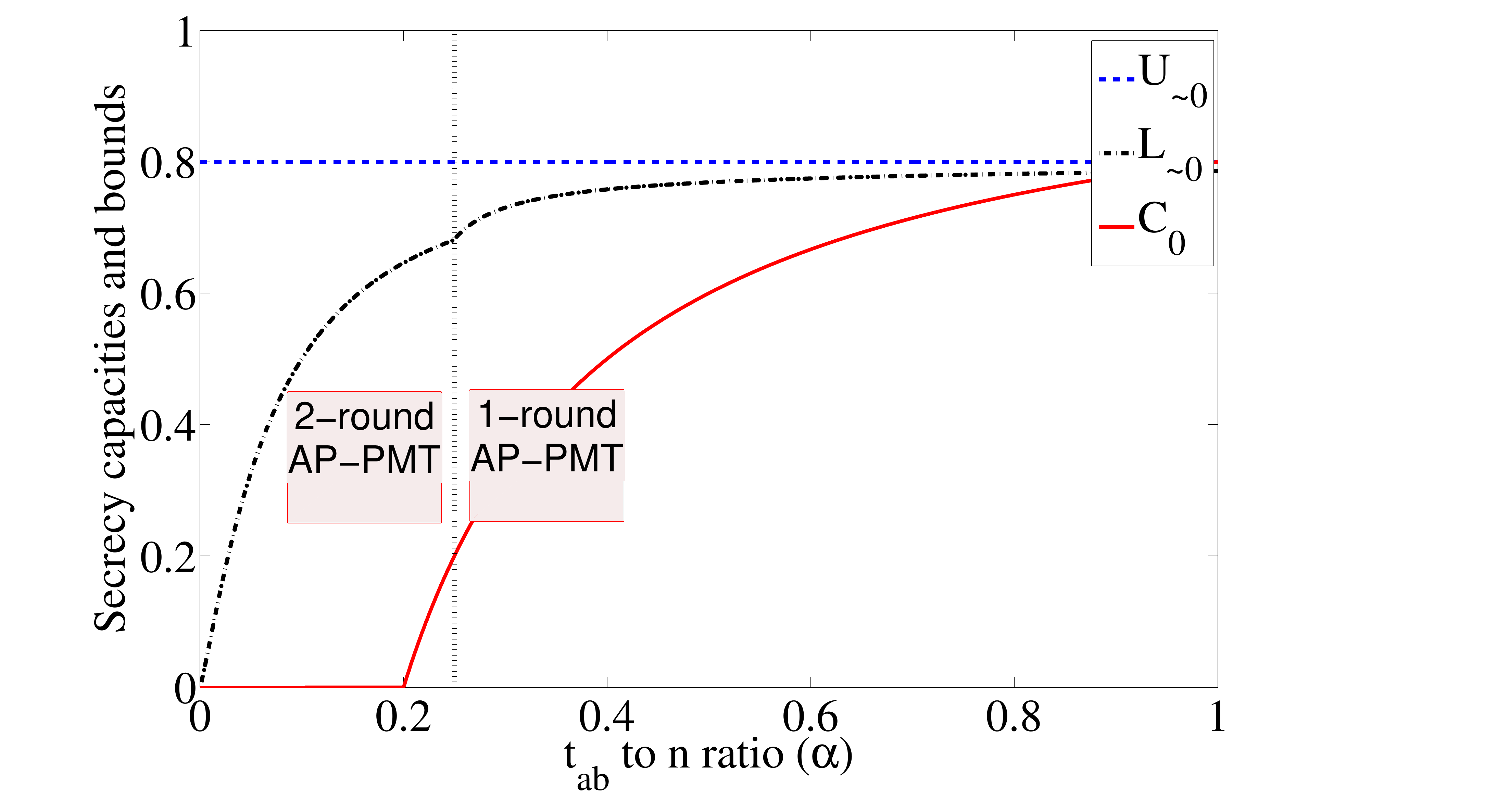}
    \label{fig-cap_tab}
  }
    \caption{Comparing the secrecy capacities and bounds.}
    \label{fig-capacities}
\end{figure}

The gap between the bounds on $C_{\sim 0}$ bridges as we move towards full access by increasing Alice's and Bob's connectivity $t_a$ and $t_b$. This is partly intuitive since in the full-access case ($t_a=t_b=n$) both P-secrecy and AP-secrecy capacities are expected to equal $1-\frac{t_e}{n}$. But what causes the gap when $t_a$ and $t_b$ are relatively small is mainly due to the two-round AP-PMT with larger communication overhead ($\xi_3$). Figure \ref{fig-cap_tab} graphs the same three quantities ($C_0$, $L_{\sim 0}$, and $U_{\sim 0}$) with respect to $\alpha = \frac{\tmin}{n}$, assuming $\len=100$, $t_a=t_b$, and $t_e = 0.2 n$. Observe that the AP-secrecy capacity upper-bound is always $U_{\sim 0} = 1- \frac{t_e}{n} = 0.8$. For $\alpha \geq 0.25$, one-round AP-PMT is possible and achieves rates quite close to $U_{\sim 0}$. For $\alpha < 0.25$, the interactive AP-PMT $\FF_3$ produces higher rates. This experiment shows we do not have a close approximation of achievable rates in the low connectivity regime, where $t_a,t_b < t_e$. Finding an answer to this question remains an open problem.

\section{Practical Consideration}\label{sec-practical}
We discuss two practical applications of our PMT results in the multipath setting model, i.e., sending secret data over (i) multiple-frequency links and (ii) multiple-route networks. Both communication scenarios include a set of paths that connect the communicants and can be tapped into by present eavesdroppers. Provided that the communication system does not allow access to all paths to the parties (in particular the eavesdropper), we hope for the possibility of PMT.

\subsection{PMT using multiple-frequency links}\label{subsec-PMT-frequency}
It is clear to see that multiple-frequency communication environments, such as wireless, realize the multipath setting described in this work. There are many non-overlapping frequency channels that can be used for signal transmission between a sender and a receiver and the communication may be intercepted by an eavesdropper tapping into the system. Our PMT results promise the possibility of secure communication here, provided that the eavesdropper does not have simultaneous access to all frequencies (i.e., $t_e < n$ in our setting). The challenge is thus to design a multiple-frequency system that makes it infeasible for the adversary to capture data over all frequency channels. Unfortunately, existing frequency-hopping solutions do not satisfy this requirement. Bluetooth for example transmits data at speed of 1Mbps over 79 adjacent 1-MHz frequency channels within the 2.4--2.48 GHz band. By the current technology, one can obtain a receiver device much smaller than a laptop to capture all 80-MHz Bluetooth range, convert it into digital, and store it to a disk at a rate in the order of 80 Mbps. This lets the eavesdropper record hours of communication in a 1 Terabyte disk.

It is yet possible to design systems that serve our purpose. It is practically infeasible to have a single transceiver device (in particular ADC) to deal with wider than $100$ MHz of the flying signal \cite{LoFe07}. All we need is to use a system whose frequency channels are far apart. Consider for instance a system design that uses $n=70$ 20-MHz frequency-channels that are distributed evenly (with 80-MHz distances) over the 57--64 GHz frequency range. This range is called the FCC unlicensed millimeter-wave band and is allocated by Canada/US for WiGig (Gigabit WiFi). Data over each channel is transmitted at the speed of at least 100 Mbps, (e.g., by using 6-bit 64-QAM modulation). Since there is only one frequency channel in each 100-MHz slot, the eavesdropper would require $70$ transceiver blocks to access all $70$ channels simultaneously. This may not be practical in certain scenarios due to device expense issues or physical space restriction (e.g., stealth attack on indoor communication). Let us assume that up to $t_e =  35$ (i.e. $n/2$) transceiver blocks can be embedded in the eavesdropper's device, while the legitimate communicants are provided with only four such blocks, implying simultaneous access to $t_a=t_b=4$ channels.

We have figured out all the multipath setting parameters, except the interval length $\len$. The eavesdropper may want to switch between the $10$-MHz frequency channels to learn more information about the communication. Fastest frequency synthesizers allow the eavesdropper to switch between these channels in about $1 \mu$s \cite{WinRadio}. Although we may allow a larger switching time for legitimate communicants, the $1 \mu$s switching time determines $\len$ in our design. At the speed of 100 Mpbs, we obtain $\len = \lfloor 10^{-6} \times 100 \times 2^{20} \rfloor = 104$ bits: The transmitter sends 104 bits over each accessed channel and then switches to a next set of channels, thus preventing the eavesdropper from switching channels during the 104-bit long interval. This example leads us to the $(70,4,4,35,104)$-multipath setting for which the two-round (interactive) AP-PMT scheme $\FF_3$ sends private data at rate $17$\%. This solution does not require pre-shared keys and provides information-theoretic security.

\subsection{PMT using multiple-route networks}\label{subsec-PMT-network}
Large networks such as sensor networks, mobile ad hoc networks (MANETs), and the Internet allow nodes to communicate over possibly multiple routes. It has been shown \cite{YeKrTr03} that when multiple routes are available, the communicants can benefit from multipath routing to enhance the reliability of transmission. We would like to study whether multipath routing can be used for privacy of communication in the case an eavesdropper taps into some intermediate router nodes. We here focus on communication over MANETs. Studies have shown that the average number of node-disjoint paths in a moderately-dense MANET (of around 500 nodes) is usually more than $10$. We consider the following scenario: There are a total of $n=10$ paths between the source node and the destination node. The source can send only over $t_a=2$ paths while the destination receives data through all $t_b=10$ paths. The adversary's resources allow for compromising at most $t_e=8$ paths at a time, and at least $1$ millisecond is needed to redirect resources to tap into new nodes (and paths); this is quite plausible, noting the technical challenges of tapping into communicating devices. The source transmits data at the speed of $512$ Kbps, implying $\len=\lfloor 10^{-3} \times 512 \times 2^{10} \rfloor =524$. This leads to the $(10,2,8,10,524)$-multipath setting for which the simplified version of scheme $\FF_2$ (with Stage (i) only -- see Remark \ref{remark-FF2}) guarantees private transmission at the highest possible rate of $20$\%.

\section{Conclusion and Future Work}
We have derived the necessary and sufficient conditions for the possibility of P-PMT and AP-PMT in the multipath setting. We also derived lower and upper bounds
on the P-secrecy and AP-secrecy capacities. Although in the full-access case ($t_a=t_b=1$) P-PMT and AP-PMT behave the same, in general, AP-PMT protocols may attain much higher secrecy rates. The maximum rate for P-PMT is $[1-\frac{t_e}{\tmin}]_+$, whereas AP-PMT protocols can achieves much higher rates close to the upper-bound $1-\frac{t_e}{n}$. The is yet a gap between the proved achievable rates and this upper-bound. \textit{Bridging the gap is an interesting question which we leave for future work}.

Any practical communication system that benefits from the diversity of communication paths can be a test case to show the feasibility our PMT results. We considered this for the real-life scenarios of communication over multiple-frequency links and multiple-route networks. In both cases, we elaborated on how to derive multipath setting parameters and use our results to provide private communication at rates $17$\% and $20\%$, respectively. Showing the possibility of keyless communication with information-theoretic privacy is interesting. \textit{A followup work can be the design of concrete protocols considering all practical and technical concerns that may have been missing in this work.}

\bibliography{Hadibibs}{}
\bibliographystyle{abbrv}

\appendix

\section{Proof of Theorem \ref{theorem-FF0}} \label{app-proof-FF0}
It is clear that the message length equals $k=r u$ and the number of communication bits is $c= n u$. This implies the secrecy rate of $R=\frac{r u}{n u} = 1 - \frac{t_e}{n}$. Perfect reliability ($\delta=0$) holds from the reconstruction property of the SSS noting that Bob receives all shares $X_i$'s. Perfect secrecy ($\epsilon=0$) follows from the secrecy property of the SSS: Eve by collecting $t_e$ shares is unqualified and so for any message distribution $S$, it holds $I(S;View_E(S))=0$, implying
\[\forall s_1,s_2\in \bset^k:~~ SD(View_E(s_1), View_E(s_2)) = 0.\]

\section{Proof of Lemma \ref{lemma-remove-delta}} \label{app-proof-remove-delta}
Since Eve is passive, full-access communicants Alice and Bob will have the same view of the communication transcript, denoted by $X$, by the $(k,c,\delta,\epsilon)$-PMT protocol $\Pi$ over the multipath setting.  Following the protocol $\Pi$, Alice calculates the key $S=f(X,Rnd_A)$ and Bob calculates its estimate $\hat{S}=g(X,Rnd_B)$ using key derivation functions $f(.,.)$ and $g(.,.)$ as well as local randomness $Rnd_A$ and $Rnd_B$, respectively.  We shall convert $\Pi$ to a new protocol $\Pi'$ that provides the same rate and has perfect reliability. We build $\Pi'$ by converting the functions $f$ and $g$ to deterministic functions, which require no local randomness.

\ihdr{Step 1: Making g(.,.) deterministic}
We write
\begin{eqnarray*}
\Pr(\hat{S}=S) &=& \Pr(f(X,Rnd_A)=g(X,Rnd_B)) \\
    &=& \sum_{x, r} \Pr(X=x) \Pr(Rnd_B=r| X=x) \Pr(f(x,Rnd_A|X=x)=g(x,r)) \\
    &\leq& \sum_{x} \Pr(X=x) \max_r \Pr(f(x,Rnd_A|X=x)=g(x,r)) \stackrel{\triangle}{=} \Pr(f(X,Rnd_A)=g(X,R^*_B(X))).
\end{eqnarray*}
We define the deterministic function $g'(X)=g(X,R^*_B(X))$. Bob will guarantee $\delta$-reliability by returning $\hat{S'}=g'(X)$ instead of $\hat{S}$, i.e., $\Pr(\hat{S'} = S) \geq \Pr(\hat{S}=S)$.

\ihdr{Step 2: Making f(.,.) deterministic}
For each $x$ as an instance of the communication transcript $X$, define the set
\begin{eqnarray*}
{\cal R}_{A}(x)=\{r ~|~ f(x,r) = g'(x) \}.
\end{eqnarray*}
The set includes randomness values that, if chosen by Alice, lead to the prefect reliability of key establishment. Perfectly reliability can be thus obtained as the event ${\EE}$ which refers to when $Rnd_A \in {\cal R}_A(X)$. Usinh the reliability property of $\Pi$, the event ${\EE}$ occurs with probability
\[\Pr({\EE}) = \Pr(\hat{S'} = S) \geq \Pr(\hat{S}=S) \geq 1-\delta.\]

Let $R^*_A(X)$ be an arbitrary (but fixed) member of the random set ${\cal R}_A(X)$. Define deterministic key derivation function $f'(X)=f(X,R^*_A(X))$. The protocol $\Pi'$ is the same as $\Pi$ except for the key derivation step, the parties return $S' = f'(X)$ and $\hat{S'} = g'(X)$.

Clearly $\Pi'$ has the same rate since the communication transcripts and the key length remain unchanged. $\Pi'$ is also perfectly-reliable since $R^*_A(X) \in {\cal R}_{A}(x)$ and thus $f'(x)=g'(x)$. It remains to show the secrecy (and uniformity) of $S'$. Let $Z \in \ZZ$ be Eve's view of the communication (in both $\Pi$ and $\Pi'$). Note that when $\EE$ holds $S=S'$ and so for any $s\in \bset^k$ and $z \in \ZZ$, we have
\begin{eqnarray*}
\Pr(S'=s, Z=z | \EE) = \Pr(S=s, Z=z | \EE).
\end{eqnarray*}
The joint-distribution of $(S',Z)$ satisfies
\begin{eqnarray*}
\Pr(S'=s,Z=z)&=& \Pr(S'=s, Z=z | \EE) \Pr(\EE) +  \Pr(S'=s, Z=z | \overline{\EE}) \Pr(\overline{\EE})\\
    &=& \Pr(S=s, Z=z | \EE) \Pr(\EE) +  \Pr(S'=s, Z=z | \overline{\EE}) \Pr(\overline{\EE})\\
    &=& \Pr(S=s, Z=z) - \Pr(S=s, Z=z | \overline{\EE}) \Pr(\overline{\EE}) +  \Pr(S'=s, Z=z | \overline{\EE}) \Pr(\overline{\EE})\\
\end{eqnarray*}
The secrecy parameter $\epsilon'$ of the new protocol is obtained as
\begin{eqnarray*}
SD([S',Z];[U_k,Z]) &=& \frac{1}{2} \sum_{s,z} |\Pr(S'=s, Z=z) - \Pr(U_k=s, Z=z)| \\
        &\leq & \frac{1}{2} \sum_{s,z} |\Pr(S=s, Z=z) - \Pr(U_k=s, Z=z)|  \\
        && \tab + \Pr(\overline{\EE})  \frac{1}{2} \sum_{s,z}  | \Pr(S'=s, Z=z | \overline{\EE})-\Pr(S=s, Z=z|\overline{\EE})| \\
        &\leq& \epsilon + \delta.
\end{eqnarray*}

\section{Proof of Lemma \ref{lemma-remove-epsilon}} \label{app-remove-epsilon}
Let $\Pi$ be a $(k,c,0,\epsilon)$-PMT protocol. Consider a ``uniform'' secret key $S\in \bset^k$ and denote the communication in each round $1\leq j\leq r$ (from one party to the other) by $X^{(j)}=(X^{(j)}_1,\dots,X^{(j)}_n)$, where $X^{(j)}_i \in \{0,1\}^{l^{(j)}_i}$ is transmitted over the $i$-th path. Alice and Bob have access to the whole communication string $X=(X^{(1)}, ..., X^{(r)})$. Eve's view of the communication is $Z=(Z^{(1)}, ..., Z^{(r)})$, where $Z^{(j)}=(Z^{(j)}_1,\dots,Z^{(j)}_n)$ is observed in round $j$ over some arbitrary paths ${\cal T}^{(j)}$ such that $|{\cal T}^{(j)}| \leq t_e$:
\begin{eqnarray*}
\forall 1\leq j \leq n,~~ \forall 1\leq i \leq n:~~ Z^{(j)}_i = \begin{cases}
X^{(j)}_i, & \mbox{if}~~ i \in {\cal T}^{(j)}\\
\bot, & \mbox{else}
\end{cases}.
\end{eqnarray*}
We use the above to calculate
\begin{eqnarray*}
H(S|Z) \leq H(X|Z) = min_{{\cal T}^{(j)}} H((X^{(j)}_i)_{1\leq j \leq r,~~ i \notin {\cal T}^{(j)}}) \leq  min_{{\cal T}^{(j)}} \sum_{j=1}^r \sum_{i \notin {\cal T}^{(j)}} l_i.
\end{eqnarray*}

We need to relate $H(S|Z)$ to the secret-key length in order to find an upper-bound on rate $\frac{k}{c}$. We use the relation between statistical distance and Shannon conditional entropy that is given by the following lemma \cite{HY10}.
\begin{lemma}\label{MI-SD-interplay}
For two variables $S \in \bset^k$ and $Z \in {\cal Z}$ with statistical distance $SD([S, Z]; [U_k , Z]) \leq \epsilon$, such that $U_k$ is a $k$-bit uniform string, it holds that $H(S|Z)\geq (1-\epsilon') k$, where $\epsilon'=1.25 \epsilon-\epsilon \log \epsilon$.
\end{lemma}
The lemma suggests $k \leq \frac{H(S|Z)}{1-\epsilon'}$. Hence, the secret-key rate
\begin{eqnarray*}
\frac{k}{c} \leq \frac{1}{1-\epsilon'} min_{{\cal T}^{(j)}}  \frac{ \sum_{j=1}^r \sum_{i \notin {\cal T}^{(j)}} l_i }{\sum_{j=1}^r \sum_{i=1}^n l_i} \leq \frac{(n-t_e)r}{nr (1-\epsilon')} = \frac{1-\frac{t_e}{n}}{1-\epsilon'}.
\end{eqnarray*}
The last inequality becomes equal (maximized) when in each round the same number of bits is transmitted over all paths, i.e., $\forall j:~~ l^{(j)}_1=\dots=l^{(j)}_n=l^{(j)}$, and this leads to the upper bound $\frac{(n-t)rl^{(1)}}{nrl^{(1)}}$.

\section{Proof of Lemma \ref{lemma-upperbound-0}} \label{app-proof-upperbound-0}
Assume $\Pi$ is an $r$-round ($r\geq 1$) $(k, c, 0, 0)$-PMT protocol over the $(n,t_a,t_b,t_e,\len)$-multipath setting ($t_e$ can be less, equal, or greater than $\tmin=\min(t_a,t_b)$). Define Alice's, Bob's, and Eve's views of the communication through $\Pi$ by
\begin{eqnarray*}
&& X=(X_i^{(j)})_{1\leq i\leq n ~,~ 1\leq j \leq r},\\
&& Y=(Y_i^{(j)})_{1\leq i\leq n ~,~ 1\leq j \leq r},\\
&& Z=(Z_i^{(j)})_{1\leq i\leq n ~,~ 1\leq j \leq r},
\end{eqnarray*}
where $V^{(j)}_i$ ($V\in \{X,Y,Z\}$) indicates the view of a party over the $i$-th path during the $j$-th round. The views are specified as follows. In the beginning each round $j$, Alice, Bob, and Eve select $\MM^{(j)}_A$ of size $t_a$, $\MM^{(j)}_B$ of size $t_b$, and ${\cal T}^{(j)}$ of size $t_e$ as their selected paths indexed from $\{1,\dots,n\}$. In any round either Alice or Bob sends information over the selected paths, then the other communicant and Eve observe information on their chosen paths. For instance, for a round $1\leq j \leq r$ with Alice as the sender, the views are

{\small
\begin{eqnarray*}
X_i^{(j)} = \begin{cases}
X_i^{(j)} & \mbox{if}~ i \in \MM^{(j)}_A\\
\bot & else
\end{cases},
~
Y_i^{(j)} = \begin{cases}
X_i^{(j)} & \mbox{if}~ i \in \MM^{(j)}_A \cap \MM^{(j)}_B \\
\bot & else
\end{cases},
~ \mbox{and}~
Z_i^{(j)} = \begin{cases}
X_i^{(j)} & \mbox{if}~ i \in \MM^{(j)}_A \cap {\cal T}^{(j)}\\
\bot & else
\end{cases}.
\end{eqnarray*}
}

Denote by ${\cal Z}={\cal Z}^{(1)} \times\dots \times {\cal Z}^{(r)}$ the set of all possible views of Eve's for all rounds. Note that in each round $1\leq j \leq r$, the set ${\cal Z}^{(j)}$ consists of all and only $n$-sequences with at least $n-t_e$ elements being equal to $\bot$. Denote by ${\cal T}_{z^{(j)}}$ the set of elements in $z^{(j)}$ that are not equal to $\bot$. Contiuing with the assumption that Alice is the sender of round $j$, we have $|{\cal T}_{z^{(j)}}| \leq \min(t_e,t_a)$ and the equality corresponds to when the Eve's captured paths have the maximum intension with those chosen by Alice. Similarly, $|{\cal T}_{z^{(j)}}| \leq \min(t_e,t_b)$ when Bob is the sender of round $j$. The perfect secrecy of the protocol $\Pi$ implies that the distribution of the message is not changed given the view of the adversary, i.e.,
\[\forall s \in \bset^k,z \in {\cal Z}:~ \Pr(S=s, Z=z)=\Pr(S=s),\]
which follows
\begin{eqnarray}\label{upper-ps}
\forall z \in {\cal Z}:~ H(S|Z=z) = H(S).
\end{eqnarray}
On the other hand, the reliability property ($\hat{S}=S$) implies $H(S|\hat{S})=0$, thus
\begin{eqnarray}\label{upper-pr}
H(S) = I(S;\hat{S}).
\end{eqnarray}
We combine (\ref{upper-ps}) and (\ref{upper-pr}) to write
\begin{eqnarray*}
H(S) &=& \min_{z \in {\cal Z}} H(S|Z=z) = \min_{z \in {\cal Z}} I(S;\hat{S}|Z=z)  \\
    &\stackrel{(a)}{\leq}& \min_{z \in {\cal Z}} I(X;Y|Z=z) \leq  \min_{z \in {\cal Z}} H(X,Y|Z=z) \\
    &\stackrel{(b)}{\leq}& \sum_{1\leq j\leq r}  \min_{z^{(j)} \in {\cal Z}^{(j)}} H(X^{(j)},Y^{(j)}|Z^{(j)}=z^{(j)}) \\
    &\leq& \sum_{1\leq j\leq r}  \min_{z^{(j)}\in {\cal Z}^{(j)}}  \sum_{i \in \MM^{(j)}_A \cap \MM^{(j)}_B \cap {\cal T}'_{z^{(j)} }} l^{(j)}_i,
\end{eqnarray*}
where ${\cal T}'_{z^{(j)}}$ is the complement of ${\cal T}_{z^{(j)}}$, and $l^{(j)}_i$ denotes the number of transmitted bits (by either Alice or Bob) over the $i$-th path in the $j$-th round. Inequality (a) is due to the Markovity of $S \leftrightarrow X \leftrightarrow Y \leftrightarrow \hat{S}$, and inequality (b) follows from the chain rule of conditional entropy. Considering uniform distribution $H(S)=k$ for the message, we can upper-bound the secrecy rate as
\begin{eqnarray*}
R= \frac{k}{c} = \frac{H(S)}{c}\leq \frac{ \sum_{j=1}^r \min_{z^{(j)}} \sum_{i \in \MM^{(j)}_A \cap \MM^{(j)}_B \cap {\cal T}'_{z^{(j)} }} l^{(j)}_i}{\sum_{j=1}^r \sum_{i \in \MM^{(j)}_A } l^{(j)}_i} \leq \frac{\sum_{i=1}^r [\tmin-t_e]_+ l^{(j)}}{\sum_{i=1}^r \tmin l^{(j)}} = [1-\frac{t_e}{\tmin}]_+,
\end{eqnarray*}
The last inequality holds since (i) there is $z^{(j)}$ that gives the maximum $|{\cal T}_{z^{(j)}}|\geq \min(t_e,\tmin)$, implying that $|\MM^{(j)}_A \cap \MM^{(j)}_B \cap {\cal T}'_{z^{(j)} }| \leq [\tmin-t_e]_+$ and (ii) the fraction becomes maximized when the same number of bits is transmitted over all paths in every round, i.e., $\forall j ~\forall i \in \MM^{(j)}_A:~ l^{(j)}_i=l^{(j)}$.

\section{Proof of Theorem \ref{theorem-FF1}} \label{app-proof-FF1}
The family of protocols starts from message length $k = r_2 \len$, obtained by choosing $q_1$ large enough such that (see (\ref{tk-FF1}))
\begin{eqnarray*}
q_1 \geq \frac{2^{\len/2}-1}{t_e} ~~~ \mbox{and}~~~ q_2 \geq \max \left( \frac{(2+\psi)n}{\psi^2 t_e}\ln\frac{1}{\epsilon} ~,~ \frac{2^{\len/2}-1}{t'_e} \right).
\end{eqnarray*}
We analyze the rate, reliability, and secrecy of the family $\FF_1$ as follows.

\bhdr{Secrecy rate}
Using (\ref{tk-FF1})-(\ref{c2-FF1}) and letting $\tmin=\min(t_a,t_b)$ and $\Delta=(2^{\frac{\len}{2}-2}-0.25)^{-1}$, we have
\begin{eqnarray*}
R_1   &=& \frac{r_2 \len}{(q_1+q_2) \tmin \len} = \frac{q_2}{q_2+q_1} \left(1-\frac{t_e'}{\tmin} - \frac{2g_2}{q_2 \tmin}\right) = \frac{q_2}{q_2+\frac{r_1}{\tmin - t_e - \frac{2g_1}{q_1}}} \left(1-\frac{t_e'}{\tmin} - \frac{2g_2}{q_2 \tmin}\right) \\
    &=& \frac{q_2}{q_2+\frac{q_2 w}{\len \tmin \left( 1 - \frac{t_e}{\tmin} - \frac{2 g_1}{q_1 \tmin}\right) }} \left(1-\frac{t_e'}{\tmin} - \frac{2g_2}{q_2 \tmin}\right) \stackrel{(a)}{\geq} \frac{1-\frac{t_e'}{\tmin} - \Delta }{1+\frac{\log{n \choose \tmin}}{\len \tmin \left( 1 - \frac{t_e}{\tmin} - \Delta \right) }} \stackrel{(b)}{\geq} \frac{1-\frac{t_e'}{\tmin} - \Delta }{1+\frac{\log(e n / \tmin)}{\len \left( 1 - \frac{t_e}{\tmin} - \Delta \right) }}.
\end{eqnarray*}

Inequality (a) follows by using a similar argument as in (\ref{Delta}), noting the choices of $g_1$ and $g_2$ (\ref{tk-FF1}) as well as $q_1 \geq \frac{2^{\len/2}-1}{t_e}$ and $q_2 \geq \frac{2^{\len/2}-1}{t'_e}$.  Inequality (b)
holds due to ${n \choose \tmin} < (ne/\tmin)^\tmin$ which follows from Stirling's approximation. The fact that $\psi>0$ can be arbitrarily small implies that $\lim_{\psi \to 0} \frac{t'_e}{\tmin} = \frac{t_e}{n}$; this completes the rate analysis.

\bhdr{Perfect-reliability} This is trivial: In both blocks (key transport and coordinated PMT), Alice and Bob use P-PMT schemes over fixed paths and hence Bob always recovers the transmitted message using the reconstruction function of the SSS.

\bhdr{$\epsilon$-secrecy} We show this considering the choice of $q_2 \geq \frac{(2+\psi)n}{\psi^2 t_e}\ln(\frac{1}{\epsilon})$. Recall that the coordinated PMT block relies on the assumption that Eve observes communication over at most $q_2 t_e'$ of the $q_2 \tmin$ paths secretly chosen by Alice and Bob for $q_2$ time intervals.  Let $T'_i \leq \min(\tmin,t_e)=t_e$ be the number of the secretly chosen paths (in the $i$-th interval of Stage (ii)) which intersect those $t_e$ paths captured by Eve. The $\epsilon$-secrecy of the protocol (see \ref{sec} in Definition \ref{def-PMT}) is directly related to the probability that $\sum_{i=1}^{q_2} T'_i > {q_2} t'$: If $\sum_{i=1}^{q_2} T'_i \leq {q_2} t'$, the SSS secrecy property guarantees no information leakage. For every $s_1,s_2 \in \bset^k$
\begin{eqnarray*}
SD(View_E(s_1),View_E(s_2)) \leq \Pr(\sum_{i=1}^{q_2} T'_i \leq {q_2} t_e') \times 0 + \Pr (\sum_{i=1}^{q_2} T'_i > {q_2} t_e') \times 1 = \Pr (\sum_{i=1}^{q_2} T'_i > {q_2} t_e').
\end{eqnarray*}
We upper-bound this probability as follows. The variables $T'_i$ follow the hyper-geometric distribution
\begin{eqnarray*}
\forall 0\leq j\leq t_e:~~ \Pr(T'_i=j) = \frac{{\tmin \choose j} {n-\tmin \choose t_e-j} }{ {n \choose t_e}},
\end{eqnarray*}
whose expected value equals $\frac{\tmin t_e}{n}$. We now consider the independent and identically distributed random variables $\frac{T'_i}{\tmin}$ which take values between 0 and 1 and have equal expected value of $\mu=t_e/n$. Applying the Chernoff bound to the sum of these independent variables \cite{Ch52} shows us the following upper-tail probability bound:
\begin{eqnarray*}
\Pr\left(\sum_{i=1}^{q_2} T'_i> {q_2} t_e' \right) = \Pr\left( \sum_{i=1}^{q_2} \frac{T'_i}{\tmin} > (1+\psi) {q_2} \mu \right) < e^{-\frac{\psi^2}{2+\psi} {q_2} \mu} \leq  e^{- \ln(1/\epsilon)} = \epsilon.
\end{eqnarray*}
The first inequality is due to the Chernoff bound and the last one holds because of the choice of ${q_2}$.


\section{Proof of Theorem \ref{theorem-FF2}} \label{app-proof-FF2}
The family of protocols starts from message length  $k=r_2 \len$, obtained by choosing $q_1$ large enough such that (see (\ref{tk-FF2}))
\begin{eqnarray*}
q_1 \geq \max \left(\frac{2 n}{\psi^2 t_b} \ln(\frac{1}{\delta}) ~,~  \frac{(2+\psi)n}{\psi^2 t_e}\ln(\frac{2}{\epsilon}) ~,~ \frac{2^{\len/2}-1}{t'_{e,1}} \right) ~~~ \mbox{and}~~~ q_2 \geq \max \left( \frac{(2+\psi)n}{\psi^2 t_e}\ln(\frac{2}{\epsilon}) ~,~ \frac{2^{\len/2}-1}{t'_{e,2}} \right).
\end{eqnarray*}

\bhdr{Secrecy rate}
Using (\ref{tk-FF2}) and letting $\Delta=(2^{\frac{\len}{2}-2}-0.25)^{-1}$, we reach the secrecy rate:
\begin{eqnarray*}
R_2   &=& \frac{r_2 \len}{(q_2 \tmin + q_1 t_a) \len} = \frac{q_2}{q_2+q_1 t_a/\tmin} \left(1-\frac{t'_{e,2}}{\tmin} - \frac{2g_2}{q_2 \tmin}\right) = \frac{q_2}{q_2+\frac{r_1 t_a}{\tmin(t'_{b,1} - t'_{e,1} - \frac{2g_1}{q_1})}} \left(1-\frac{t'_{e,2}}{\tmin} - \frac{2g_2}{q_2 \tmin}\right) \\
    &=& \frac{q_2}{q_2+\frac{q_2 w}{\len \tmin \left( \frac{t'_{b,1}-t'_{e,1}}{t_a} - \frac{2 g_1}{q_1 t_a}\right) }} \left(1-\frac{t'_{e,2}}{\tmin} - \frac{2g_2}{q_2 \tmin}\right) \stackrel{(a)}{\geq} \frac{1-\frac{t'_{e,2}}{\tmin} - \Delta }{1+\frac{\log{n \choose \tmin}}{\len \tmin \left( \frac{t'_{b,1}-t'_{e,1}}{t_a} - \Delta \right) }} \stackrel{(b)}{\geq} \frac{1-\frac{t'_{e,2}}{\tmin} - \Delta }{1+\frac{\log(e n / \tmin)}{\len \left( \frac{t'_{b,1}-t'_{e,1}}{t_a} - \Delta \right) }}\\
    &=& \frac{1- (1+\psi)\frac{t_e}{n} - \Delta }{1+\xi_{2,\psi}},~~~ \mbox{where}~~~ \xi_{2,\psi} = \frac{\log(e n / \tmin)}{\len \left( \frac{(1-\psi) t_b - (1+\psi) t_e}{n} - \Delta \right)}.
\end{eqnarray*}
Inequality (a) follows from the definition of $g_1$ and $g_2$ (\ref{tk-FF2}) and the choices of $q_1 \geq \frac{2^{\len/2}-1}{t'_{e,1}}$ and $q_2 \geq \frac{2^{\len/2}-1}{t'_{e,2}}$.  Inequality (b)  holds due to ${n \choose \tmin} < (ne/\tmin)^\tmin$ which follows from Stirling's approximation. Noting that $\psi>0$ can be made arbitrarily small completes the proof of rate.

\bhdr{$\delta$-reliability} The protocol may fail ($\hat{S} \neq S$) only if Bob cannot recover the transported key $W$ at Stage (i), whose probability is upper-bounded as follows. Let $T'_i \leq \tmin$ denote the number of elements that Bob receives in time interval $1 \leq i \leq {q_1}$. Note from the description of the scheme that Bob fails in obtaining $W$ only if he receives less than $q_1 t'_{b,1}$ elements, i.e., when $\sum_{i=1}^{q_1} T'_i < q_1 t'_{b,1}$. For all $1\leq i \leq q_1$, the random variable $T'_i$ follows the hyper-geometric distribution
\[\forall 0 \leq j \leq \tmin:~~ \Pr(T'_i = j) = \frac{{t_a \choose j} {n-t_a \choose t_b-j}}{{n\choose t_b}},\]
which has an expected value of $E(T'_i) = \frac{t_a t_b}{n}$. This implies that normalized independent random variables $\frac{T'_i}{t_a}$ (for $1 \leq i \leq q_1$) take values between 0 and 1 and have the expected value $\mu = t_b/n$. We use the Chenroff  bound \cite{Ch52} on the sum of independent normalized random variables to obtain (note that $t'_{b,1}=\frac{(1-\psi)t_a t_b}{n}$)
\begin{eqnarray*}
\Pr(\hat{S} \neq S) &\leq& \Pr(\sum_{i=1}^{q_1} T'_i < {q_1} t'_{b,1}) = \Pr(\sum_{i=1}^{q_1} \frac{T'_i}{t_a} < (1-\psi) {q_1} \mu) \stackrel{(a)}< e^{-\frac{\psi^2}{2} {q_1} \mu} \stackrel{(b)}\leq e^{- \ln(1/\delta)} = \delta.
\end{eqnarray*}
Inequality (a) is due to the Chernoff bound and inequality (b) is due to the choice of $q_1 \geq \frac{2 n}{\psi^2 t_b}\ln(\frac{1}{\delta})$.

\bhdr{$\epsilon$-secrecy} Thanks to the secrecy of the quasi-ramp SSS, the protocol provides ``perfect secrecy'' if Eve receives at most $q_1 t'_{e,1}$ share elements in the key-transport block and at most $q_2 t'_{e,2}$ share elements during the coordinated PMT. We show that both above assumptions holds except with probability $\epsilon$. Let the random variable $T''_i$ denote the number of paths that are accessed by both Alice and Eve during interval $1\leq i \leq q_1+q_2$ (this includes both stages of the protocol). For the key-transport (intervals $1$ to $q_1$), the distribution of $T''_i$ is hyper-geometric with an expected value of $\frac{t_a t_e}{n}$. The normalized variables $\frac{T''_i}{t_a}$ (for $1\leq i \leq q_1$) are independent, take values between 0 and 1, and each has the expected value $\mu_e=\frac{t_e}{n}$. Applying the Chernoff bound to the sum of these variables, we reach at (note that $t'_{e,1}=(1+\psi)\frac{t_a t_e}{n}$ and $q_1 \geq \frac{(2+\psi)n}{\psi^2 t_e}\ln(\frac{2}{\epsilon})$)
\begin{eqnarray*}
\Pr\left(\sum_{i=1}^{q_1} T''_i > q_1 t'_{e,1} \right) = \Pr\left( \sum_{i=1}^{q_1} \frac{T''_i}{t_a} > (1+\psi) {q_1} \mu_e \right) < e^{-\frac{\psi^2}{2+\psi} {q_1} \mu_e} \leq  e^{- \ln(1/\epsilon)} = \frac{\epsilon}{2},
\end{eqnarray*}
Similarly, we have independent and normalized variables $\frac{T''_i}{\tmin}$ for $q_1+1\leq i \leq q_1+q_2$ with the expected values of $\mu_e=\frac{t_e}{n}$. Using the Chernoff inequality results in (note that $t'_{e,2}=(1+\psi)\frac{\tmin t_e}{n}$ and $q_2 \geq \frac{(2+\psi)n}{\psi^2 t_e}\ln(\frac{2}{\epsilon})$)
\begin{eqnarray*}
\Pr\left(\sum_{i=q_1}^{q_1+q_2} T''_i > q_2 t'_{e,2} \right) < \frac{\epsilon}{2},
\end{eqnarray*}

The $\epsilon$-secrecy of $\FF_2$ (see \ref{sec} in Definition \ref{def-PMT}) follows from the above: For every $s_1,s_2 \in \bset^k$,
\begin{eqnarray*}
SD(View_E(s_1),View_E(s_2)) \leq \Pr (\sum_{i=1}^{q_1}  T''_i > q_1 t'_{e,1} \cup \sum_{i=q_1+1}^{q_1+q_2}  T''_i > q_2 t'_{e,2}) \times 1 < \epsilon.
\end{eqnarray*}

\section{Proof of Proposition \ref{prop-ow-t_e>mb}} \label{app-proof-ow-t_e>mb}
Let $\Pi$ be any one-round $(R, \delta,\epsilon)$-PMT protocol and let $S \in \bset^k$ be a uniform message to be delivered by $\Pi$. The protocol is generally described by two public distributions $P_{Rnd_A}$ and $P_{Rnd_B}$ as well as three public (deterministic) functions $f$, $g$, and $h$ which may be bounded according to the multipath setting constraints. Alice and Bob first generate random strings $Rnd_A$ and $Rnd_B$ independently and according to distributions $P_{Rnd_A}$ and $P_{Rnd_B}$, respectively. Alice sends $\overline{X}=f(Rnd_A,S)$ over the multipath setting via possibly multiple time intervals. Bob creates his view of Alice's communication through $\overline{Y}=g(Rnd_B, \overline{X})$, where $g$ is bounded by Bob's $t_b$-path access limitation in each time interval and retrieves the message as $\hat{S}=h(Rnd_B,\overline{Y})$. With the public knowledge of $P_{Rnd_B}$, $g$, and $h$, Eve can follow Bob's strategy: She generates $Rnd_E$ according to the distribution $P_{Rnd_B}$ and then calculates $S' = h(Rnd_E,g(Rnd_E,\overline{X}))$. Note that Eve can apply the functions $g$ and $h$ since she has ability of access to $t_e\geq t_b$ paths in a time interval. From the properties of the statistical distance, it follows that
\begin{eqnarray*}
SD([S,\overline{Y}];[S,\overline{Z}]) &=& SD([S,g(Rnd_B, f(Rnd_A,S))];[S,g(Rnd_E, f(Rnd_A,S))]) \\
    &\leq& SD([S,Rnd_B,Rnd_A];[S,Rnd_E,Rnd_A]) = 0.
\end{eqnarray*}
The inequality holds since functions $g$ and $f$ are deterministic and the last equality is true since $S$, $Rnd_A$, $Rnd_B$, and $Rnd_E$ are all independent and $Rnd_E$ follows the  distribution of $Rnd_B$. Similarly, one can show
\begin{eqnarray*}
SD([U_k,\overline{Y}];[U_k,\overline{Z}]) = 0,
\end{eqnarray*}
where $U_k$ is a $k$-bit fresh uniform random string. We use the triangle inequality property of statistical distance to reach
\begin{eqnarray*}
SD([S,\overline{Y}];[U_k,\overline{Y}]) &\leq& SD([S,\overline{Y}];[S,\overline{Z}]) + SD([S,\overline{Z}];[U_k,\overline{Z}]) + SD([U_k,\overline{Z}];[U_k,\overline{Y}])  \\
    &=& SD([S,\overline{Z}];[U_k,\overline{Z}]) \leq \epsilon,
\end{eqnarray*}
where the last inequality follows from the $\epsilon$-secrecy of $\Pi$. From the inter-relation between mutual information and statistical distance (see Lemma \ref{MI-SD-interplay} in Appendix \ref{app-remove-epsilon}), we reach the following lower bound on Bob's uncertainty about the uniform message $S$:
\begin{eqnarray*}
H(S|\hat{S}) \geq H(S|\overline{Y}) = k - I(S; \overline{Y}) \geq k - 2\epsilon (\frac{k}{R}+\log(\frac{1}{\epsilon})) \geq k (1- \frac{2\epsilon}{R}) - 2.
\end{eqnarray*}
Combining the $\delta$-reliability of $\Pi$ with the inverse of Fano's inequality proves
\begin{eqnarray*}
\delta \geq \Pr(S\neq \hat{S}) \geq \frac{H(S|\hat{S})-1}{k} \geq 1 - \frac{2\epsilon}{R} - \frac{3}{k},
\end{eqnarray*}
which implies $R \leq \frac{2\epsilon}{1-\delta-\alpha}$, noting that $k \geq 3/\alpha$. By letting $\epsilon$ and $\delta$ be arbitrarily small, we obtain
\[ \lim_{\epsilon,\delta \to 0} R = 0.\]

\section{Proof of Theorem \ref{theorem-FF3}} \label{app-proof-FF3}
The family of protocols starts from message length  $k=r_2 \len$, obtained by choosing $q_1$ large enough such that (see (\ref{tk-FF3}))
\begin{eqnarray*}
q_1 \geq \max \left(\frac{2 n}{\psi^2 t_a} \ln(\frac{1}{\delta})  ~,~  \frac{(2+\psi)n^2}{\psi^2 t_a t_e}\ln(\frac{1}{\epsilon}) ~,~ \frac{2^{\len/2}-1}{t'_{e,1}} \right) ~~ \mbox{and} ~~ q_2 \geq \max \left(  \frac{(2+\psi)n}{\psi^2 t_e}\ln(\frac{1}{\epsilon})  ~,~ \frac{2^{\len/2}-1}{t'_{e,2}} \right).
\end{eqnarray*}

\bhdr{Secrecy rate}
Using (\ref{tk-FF3})-(\ref{c2-FF3}) and letting $\Delta=(2^{\frac{\len}{2}-2}-0.25)^{-1}$, we calculate the secrecy rate of $\FF_3$ as:
\begin{eqnarray*}
R_3   &=& \frac{r_2 \len}{(q_2 \tmin + q_1 t_b) \len + q_1 w_1} = \frac{q_2}{q_2+q_1 \left( \frac{t_b}{\tmin} + \frac{w_1}{\tmin \len} \right)} \left(1-\frac{t'_{e,2}}{\tmin} - \frac{2g_2}{q_2 \tmin}\right) \\
    &=& \frac{q_2}{q_2+\frac{r_1}{t'_{a,1} - t'_{e,1} - \frac{2g_1}{q_1}} \left( \frac{t_b}{\tmin} + \frac{w_1}{\tmin \len} \right) } \left(1-\frac{t'_{e,2}}{\tmin} - \frac{2g_2}{q_2 \tmin}\right)
    = \frac{1-\frac{t'_{e,2}}{\tmin} - \frac{2g_2}{q_2 \tmin}}{1+\frac{w_2}{\len \left(t'_{a,1} - t'_{e,1} - \frac{2g_1}{q_1} \right)} \left( \frac{t_b}{\tmin} + \frac{w_1}{\tmin \len} \right) } \\
    &=& \frac{1-\frac{t'_{e,2}}{\tmin} - \frac{2g_2}{q_2 \tmin}}{1+\frac{\log{n \choose \tmin}}{\len t'_{a,1} \left(1 - \frac{t'_{e,1}}{t'_{a,1}} - \frac{2g_1}{q_1 t'_{a,1}} \right)} \left( \frac{t_b}{\tmin} + \frac{\log{n \choose t'_{a,1}}}{\tmin \len} \right) }
    \stackrel{(a)}{\geq}  \frac{1-\frac{t'_{e,2}}{\tmin} - \Delta}{1+\frac{\log{n \choose \tmin}}{\len t'_{a,1} \left(1 - \frac{t'_{e,1}}{t'_{a,1}} - \Delta \right)} \left( \frac{t_b}{\tmin} + \frac{\log{n \choose t'_{a,1}}}{\tmin \len} \right) } \\
    &\stackrel{(b)}{\geq}& \frac{1-\frac{t'_{e,2}}{\tmin} - \Delta}{1+\frac{\log(e n / \tmin) \left( t_b / t'_{a,1} + \log(e n / t'_{a,1}) / \len \right)}{\len \left(1 - \frac{t'_{e,1}}{t'_{a,1}} - \Delta \right)} }
    =   \frac{1-(1+\psi)\frac{t_e}{\tmin} - \Delta}{1+ \xi_{3,\psi}}, ~~~ \mbox{where}~~~\\
    &&\tab\tab \xi_{3,\psi} = \frac{\log\frac{e n}{\tmin} \left( \frac{n}{(1-\psi)t_a} + \frac{\log(e n^2 / ((1-\psi) t_a t_b) } {\len} \right)}{\len \left(1 - \frac{(1+\psi) t_e}{n} - \Delta. \right)}
\end{eqnarray*}
Inequality (a) follows from the definition of $g_1$ and $g_2$ (\ref{tk-FF3}) and the choices of $q_1 \geq \frac{2^{\len/2}-1}{t'_{e,1}}$ and $q_2 \geq \frac{2^{\len/2}-1}{t'_{e,2}}$.  Inequality (b)  holds due to Stirling's approximation. The fact that $\psi>0$ can be made arbitrarily small completes the rate analysis.

\bhdr{$\delta$-reliability} We prove the upper bound $\delta$ on the failure probability of $\FF_3$ noting that the protocol fails only if Alice receives less than $q_1 t'_{a,1}$ field elements in Round 1 of the key-agreement block. This probability  is upper-bounded as follows. Let $T'_i \leq \tmin$ be the number of elements (transmitted by Bob) that Alice receives in time interval $i$. For all $1 \leq i \leq q_1$, the random variable $T'_i$ follows the hyper-geometric distribution
\[\forall 0 \leq j \leq \tmin:~~ \Pr(T'_i = j) = \frac{{t_b \choose j} {n-t_b \choose t_a-j}}{{n\choose t_a}},\]
which has an expected value of $E(T'_i) = \frac{t_a t_b}{n}$. This implies that normalized independent random variables $\frac{T'_i}{t_b}$ (for $1 \leq i \leq {q_1}$) have the expected value of $\mu_a = t_a/n$. We use the Chernoff  bound on the sum of independent normalized random variables to show (note that $t'_{a,1}=\frac{(1-\psi)t_a t_b}{n}$ and $q_1 \geq \frac{2 n}{\psi^2 t_a}\ln(\frac{1}{\delta})$)
\begin{eqnarray*}
\Pr(\hat{S} \neq S) \leq \Pr(r < {q_1} t'_{a,1}) = \Pr(\sum_{i=1}^{q_1} \frac{T'_i}{t_b} < (1-\psi) {q_1} \mu_a) < e^{-\frac{\psi^2}{2} {q_1} \mu_a} \leq  e^{- \ln(1/\delta)} = \delta.
\end{eqnarray*}

\bhdr{$\epsilon$-secrecy} Because of the use of quasi-ramp SSS, the protocol provides ``perfect secrecy'' assuming that Eve receives $\leq q_1 t'_{e,1}$ share elements during key transport and $\leq q_2 t'_{e,2}$ share elements during coordinated PMT. In the former, secrecy is due to the $(q_1 t'_{a,1}-2g_1, r_1, g_1, q_1 t'_{a,1})$-quasi-ramp SSS which guarantees that any $q_1 t'_{a,1}-2g_1-r_1 = q_1 t'_{e,1}$ shares of $X_A$ leak no information about the key $W$. The latter is due to the $(q_2 \tmin -2g_2, r_2, g_2, q_2 \tmin)$-quasi-ramp SSS which promises message secrecy against collation of $q_2 t'_{a,1}- 2g_2 - r_2 = q_2 t'_{e,1}$ many shares.

We show that the both above assumptions hold with except with probability $\epsilon$. Let $T''_i \leq \min(t_e, T'_i)$ be the number of paths that are observed at by both Eve and Alice in interval $1\leq i \leq q_1$. Given $T'_i=j$, the probability distribution of $T''_i$ is hyper-geometric:
\[\forall 0 \leq l \leq \min(t_e,j):~~ \Pr(T''_i = l) = \frac{{j \choose l} {n-j \choose t_e-l}}{{n\choose t_e}},\]
with expected value of $\frac{t_e j}{n}$. The expected value of $T''_i$ (unconditioned on $T'_i$) is obtained as
\begin{eqnarray*}
E(T''_i) = \sum_{j=0}^{\tmin} \Pr(T'_i=j) \frac{t_e j}{n} = \frac{t_e}{n} E(T'_i) = \frac{t_e t_a t_b}{n^2}.
\end{eqnarray*}
The normalized variables $\frac{T'_i}{t_b}$ (for $1\leq i \leq {q_1}$) are independent and have the expected value of $\mu_e=\frac{t_e t_a}{n^2}$. Applying the Chernoff bound to the sum of these independent normalized variables \cite{Ch52}, we reach at (note that $t'_{e,1}=(1+\psi)t'_{a,1} t_e/n$ and $q_1 \geq \frac{(2+\psi)n^2}{\psi^2 t_a t_e}\ln(\frac{2}{\epsilon})$)
\begin{eqnarray*}
\Pr\left(\sum_{i=1}^{q_1} T''_i > q_1 t'_{e,1} \right) = \Pr\left( \sum_{i=1}^{q_1} \frac{T''_i}{t_b} > (1+\psi)(1-\psi) q_1 \mu_e \right) < e^{-\frac{\psi^2}{2+\psi} q_1 \mu_e} \leq  e^{- \ln(2/\epsilon)} = \frac{\epsilon}{2}.
\end{eqnarray*}
Similarly, letting $T''_i$ (for $q_1+1 \leq i \leq q_1+q_2$) be Eve's observed share elements during coordinated PMT, the normalized variables $\frac{T''_i}{\tmin}$ are independent and have expected value $\mu'_e=\frac{t_e}{n}$. The probability that she receives more than $q_2 t'_{e,2}$ such shares during $q_2$ intervals at Stage (ii) is upper bounded as $t'_{e,2}=(1+\psi) \tmin t_e/n$ and $q_2 \geq \frac{(2+\psi)n}{\psi^2 t_e}\ln(\frac{2}{\epsilon})$)
\begin{eqnarray*}
\Pr\left(\sum_{i=q_1+1}^{q_1+q_2} T''_i > q_2 t'_{e,2} \right) < e^{-\frac{\psi^2}{2+\psi} q_2 \mu_e} \leq \frac{\epsilon}{2}.
\end{eqnarray*}

The $\epsilon$-secrecy of $\FF_3$ (see \ref{sec} in Definition \ref{def-PMT}) follows since for every $s_1,s_2 \in \bset^k$,
\begin{eqnarray*}
SD(View_E(s_1),View_E(s_2)) \leq \Pr (\sum_{i=1}^{q_1}  T''_i > q_1 t'_{e,1} \cup \sum_{i=q_1+1}^{q_1+q_2}  T''_i > q_2 t'_{e,2}) \times 1 \leq \epsilon.
\end{eqnarray*}

\end{document}